\documentclass[11pt]{article}
\usepackage{textcomp}

\usepackage{authblk} % author style

\def\BibTeX{{\rm B\kern-.05em{\sc i\kern-.025em b}\kern-.08em
    T\kern-.1667em\lower.7ex\hbox{E}\kern-.125emX}}

\usepackage{amsmath}
\usepackage{amsfonts}
\usepackage{amssymb}
\usepackage{amsthm}
\usepackage{xcolor}
\renewenvironment{proof}[1][\proofname]{\noindent {\bfseries #1.} }{\qed }
\usepackage{upgreek}
\usepackage{graphicx}
\usepackage{subfig}
\graphicspath{ {images/} }
\usepackage{caption}
\usepackage{pifont}
\usepackage[ruled, vlined, linesnumbered]{algorithm2e}
\usepackage{xstring}
\usepackage{cite}
\usepackage{multirow}
\usepackage{calc}
\usepackage{bbm}
\newtheorem{definition}{Definition}

\newtheorem{theorem}{Theorem}
\newtheorem{lemma}[theorem]{Lemma}

\newtheorem*{remark}{Remark}
\newcommand{\bv}[1]{\mathbf{#1}}		% Bold variable for English letters
\newcommand{\bvgrk}[1]{{\boldsymbol{#1}}}	% Bold variable for Greek letters

\newcommand{\vc}{\bv{c}}

\newcommand{\vw}{\bv{w}}

\newcommand{\vA}{\bv{A}}

\newcommand{\vC}{\bv{C}}

\newcommand{\vE}{\bv{E}}
\newcommand{\vF}{\bv{F}}

\newcommand{\vI}{\bv{I}}
\newcommand{\vJ}{\bv{J}}

\newcommand{\vM}{\bv{M}}

\newcommand{\vO}{\bv{O}}
\newcommand{\vP}{\bv{P}}
\newcommand{\vQ}{\bv{Q}}

\newcommand{\vT}{\bv{T}}
\newcommand{\vU}{\bv{U}}
\newcommand{\vV}{\bv{V}}
\newcommand{\vW}{\bv{W}}
\newcommand{\vX}{\bv{X}}

\newcommand{\vDelta}{\bvgrk{\Delta}}

\newcommand{\vTheta}{\bvgrk{\Theta}}

\newcommand{\vpi}{\bvgrk{\uppi}}
\newcommand{\vPi}{\bvgrk{\Pi}}

\newcommand{\vSigma}{\bvgrk{\Sigma}}

\newcommand{\vphi}{\bvgrk{\upphi}}

\newcommand{\norm}[1]{\|{#1} \|}

\newcommand{\curlybracketsbig}[1]{\left\{ #1 \right\}}
\newcommand{\curlybrackets}[1]{\{ #1 \}}

\newcommand{\squarebrackets}[1]{[ #1 ]}
\newcommand{\parenthesesbig}[1]{\left( #1 \right)}

\newcommand{\indicator}{\mathbbm{1}}

\newcommand{\wht}[1]{\widehat{#1}}		% Widehat
\newcommand{\wtd}[1]{\widetilde{#1}}	% Widetilde

\newcommand{\R}{\mathbb{R}}				% Real line R
\newcommand{\dm}[2]
{
	\IfStrEq{#2}{1}{\R^{#1}}{\R^{#1 \x #2}}
}

\newcommand{\N}{\mathcal{N}}			% Normal Distribution N
\newcommand{\onevec}{\mathbf{1}} 		% Bold 1, All 1 vector

\newcommand{\tr}{\textup{\textbf{tr}}} 			% Trace
\newcommand{\diag}{\textup{\textbf{diag}}} 		% Diag
\newcommand{\expctn}{\mathbb{E}} 	 	% Expectation
\newcommand{\prob}{\textup{P}} 				% The upright P to denote probability
\makeatletter
\newcommand*{\T}{{\mathpalette\@transpose{}}}
\newcommand*{\@transpose}[2]{\raisebox{\depth}{$\m@th#1\intercal$}}
\makeatother

\newcommand*{\x}{\mathsf{x}\mskip1mu} 	

\newcommand{\splitatcommas}[1]{%
	\begingroup
	\begingroup\lccode`~=`, \lowercase{\endgroup
		\edef~{\mathchar\the\mathcode`, \penalty0 \noexpand\hspace{0pt plus .1em}}%
	}\mathcode`,="8000 #1%
	\endgroup
}

\newcommand{\Item}[1]{%
	\ifx\relax#1\relax  \item \else \item[#1] \fi
	\abovedisplayskip=0pt\abovedisplayshortskip=0pt~\vspace*{-\baselineskip}
} 

\title{Mode Clustering for Markov Jump Systems
\thanks{N. Ozay and Z. Du were supported by ONR grant N00014-18-1-2501, L. Balzano and Z. Du were supported by AFOSR YIP award FA9550-19-1-0026, and L. Balzano was supported by AFOSR YIP award FA9550-19-1-0026, NSF BIGDATA award IIS-1838179, and NSF CAREER award CCF-1845076.}
}

\author[ ]{Zhe Du}
\author[ ]{Necmiye Ozay}
\author[ ]{Laura Balzano}

\affil[ ]{\textit{Department of Electrical Engineering and Computer Science}}
\affil[ ]{\textit{University of Michigan}}
\affil[ ]{Ann Arbor, USA}
\affil[ ]{{\{zhedu,necmiye,girasole\}@umich.edu}}

\date{}

\begin{document}
\maketitle

\begin{abstract}
In this work, we consider the problem of mode clustering in Markov jump models. This model class consists of multiple dynamical modes with a switching sequence that determines how the system switches between them over time. Under different active modes, the observations can have different characteristics. Given the observations only and without knowing the mode sequence, the goal is to cluster the modes based on their transition distributions in the Markov chain to find a reduced-rank Markov matrix that is embedded in the original Markov chain. Our approach involves mode sequence estimation, mode clustering and reduced-rank model estimation, where mode clustering is achieved by applying the singular value decomposition and k-means. We show that, under certain conditions, the clustering error can be bounded, and the reduced-rank Markov chain is a good approximation to the original Markov chain. Through simulations, we show the efficacy of our approach and the application of our approach to real world scenarios.
\end{abstract}

\section{Introduction}
Modeling dynamic systems has been a problem of great interest in the signal processing and control communities for decades. Many real-world phenomena cannot be described with one dynamical model, and so switched models wherein the dynamics transition between different system models have been studied and applied widely. In human-made systems, for example, a robot may have different dynamics under different battery levels or when different modules within the robot fail. In nature, the temperature and humidity level will have different fluctuations under different weather conditions; brain electricity signals will behave differently under different emotions of the test subject. Note that in all these examples, the modes can switch over time. To model this switching, one systematic and probabilistic way is to assume the mode switching follows a Markov chain where future modes do not depend on past modes given the most recent mode. This Markov jump model \cite{gupta2009networked, shi2015survey} has been used in power systems, air traffic management, economics, and communication systems \cite{ugrinovskii2005decentralized, loparo1990probabilistic, liu2011probabilistic, gopalakrishnan2017comparative, svensson2008optimal, zhu2017hmm}.

A key challenge for such models is the model compactness -- how does one represent such a complicated dynamical system with as simple a model as possible? 
For example, modes like weather conditions and human emotions have extremely complex underlying dynamics with strong correlations over time. To satisfy the Markov property, one may concatenate underlying modes into a single Markov state, and Markov chains built in this way will have a state space that grows exponentially with the number of modes concatenated in the sequence. The same exponential growth rate applies when one models human-made systems with multiple sub-modules that each have multiple behavior modes (normal/abnormal). Allowing the Markov model to get extremely large is computationally inefficient for analysis and control.

Prior work studying model reduction of Markov jump models does not consider reduction of discrete state space, i.e. (reduction of number of Markovian modes), and prior work in state space reduction of Markov chain does not further consider Markov jump models. There have been several works studying the aggregation of states for Markov chains, which mainly relies on assumptions such as strong/weak lumpability, or aggregatibility properties of a Markov chain \cite{sanders2017clustering, ganguly2014markov, deng2009information, zhang2018state}. There is therefore significant potential in applying the abundant algorithms and theory in Markov chain aggregation to Markov jump systems. This can achieve model reduction from a new perspective and will benefit the analysis and control of, especially large, systems.

The work presented here addresses this gap. We observe that often times certain modes have similar transition behaviors, and these correlations between the modes can be exploited to construct a reduced-order model. By doing so, one may gain more insight into the nature of the complex model. Moreover, we will have fewer parameters to estimate or fewer control variables to design when learning and controlling the model, thus this may significantly reduce the computation burden.
We are interested in situations where the bottleneck is due to a large discrete state-space (i.e., large number of modes) and aim to cluster and aggregate the modes for reduction. We achieve this model aggregation by clustering the modes with similar transition distributions together. We assume the dynamics for each mode are known, but we have no knowledge of the true mode sequence. In our approach, we cluster based on a reduced-dimension representation of the empirical Markov transition matrix. We then re-estimate the empirical Markov matrix using this cluster information, giving us a final low-rank estimate. We discuss our method's computational advantage, and we show our approach has guaranteed performance in the sense that the clustering error and difference between reduced model and the true model can be upper bounded. Experiments show the efficacy of our approach as well as how the performance scales with the problem complexity.

\subsection{Prior Work}Previous work on Markov jump systems includes: analysis of stability and stabilization \cite{zhang2009stability}, analysis of system with time delays \cite{zhang2008analysis}, optimal control \cite{costa1995discrete}, robust control \cite{dong2008robust}, $\mathcal{H}_\infty$ filtering \cite{gonccalves2009cal}, etc. In the context of model reduction, prior work mainly focuses on the reduction of continuous state-space (or observation space): \cite{zhang2003h} studies the $\mathcal{H}_\infty$ model reduction and derives conditions under which a reduced order system can be obtained via linear matrix inequalities; \cite{kotsalis2006model} reduces the model order with the help of generalized dissipation inequalities and storage functions; \cite{kotsalis2010balanced} proposes a balanced truncation algorithm to reduce model order and gives upper bound on approximation error. While, to the best of our knowledge, the reduction of discrete state-space (number of modes) for Markov jump systems has not been considered before. 

\section{Problem Formulation} \label{sec_prelim}
\subsection{Notation}
In this paper, boldface and uppercase (lowercase) letters denote matrices (vectors); plain letters denote scalars. If $\vA$ is a matrix, then $\vA(i,j)$ indexes the $(i,j)$th element in $\vA$ and $\vA(i,j{:}k)$ indexes the row vector corresponding to the $i$th row and column $j$ through $k$. $\vA(i,:)$ indexes the $i$th row of $\vA$. Norms without subscript, i.e. $\norm{\cdot}$, all denote the $\ell_2$-norm. We let $[n]:=\curlybrackets{1,2,\dots,n}$ and $X_{0:N}:=\curlybrackets{X_i}_{i=0}^N$.

For Markov chain with state space $[n]$ and row stochastic transition matrix $\vP \in \dm{n}{n}$, we let $\vpi \in \dm{n}{1}$ denote the stationary distribution vector of $\vP$, i.e. $\vpi^\T \vP = \vpi^\T$. Furthermore, we let $\pi_{\max}:=\max_i \vpi_i, \pi_{\min}:=\min_i \vpi_i$. If $\vP$ is ergodic, then $\vpi$ is unique and $\pi_{\min}>0$. Let $\vpi_t \in \dm{n}{1}$ denote the transient state distribution of $\vP$ and $\vpi_t^\T = \vpi_{t-1}^\T \vP $. We denote with $\curlybrackets{\Omega_1, \dots, \Omega_r}$ a partition of the state space $[n]$, where each $\Omega_k$ denotes a cluster of states. We let $\Omega_{(i)}$ denote the cluster with $i$th largest cardinality.

\subsection{Preliminaries} \label{subsec_1}
The Markov switched model we consider has the following form:
\begin{gather}
	y_t = \sum_{i=1}^{n_a} a_i(X_t) y_{t-i} + \sum_{j=1}^{n_c} c_j(X_t) u_{t-j} + n_t , \\
	X_{0:N} \in [n]^{N+1} \sim \text{Markov chain} (\mathbf{P}),
\end{gather}
where $y_t, u_t, n_t$ are scalars and represent the model output, input and noise at time $t$ respectively. And $y_t$ depends on $\curlybrackets{y_{t-i}}_{i=1}^{n_a}, \curlybrackets{u_{t-j}}_{j=1}^{n_c} $ linearly through the parameters $\curlybrackets{a_i(X_t)}_{i=1}^{n_a}, \curlybrackets{c_j(X_t)}_{j=1}^{n_c}$ from mode $X_t$ at time $t$. There are $n$ modes in total and the mode sequence $X_{0:N}$ is assumed to follow a Markov chain with row stochastic Markov matrix $\vP \in \dm{n}{n}$. The initial state distribution $\vpi_0$ can be arbitrary. Note that one can omit input $u_t$ by taking $n_c=0$, which corresponds to an autonomous model. If we let 
\begin{gather}
\vw_{X_t} := [a_1(X_t), \dots, a_{n_a}(X_t), c_1(X_t), \dots, c_{n_c}(X_t)]^\T , \\
\vphi_t := [y_{t-1}, \dots, y_{t-n_a}, u_{t-1}, \dots, u_{t-n_c}]^\T,
\end{gather}
then we obtain a simpler representation of the model:
\begin{equation}\label{eq_48}
	y_t = \vw_{X_t}^\T \vphi_t + n_t, \\
\end{equation}
where the pair $\curlybrackets{y_t, \vphi_t}$ can be viewed as the observation/data.

Furthermore, we assume the Markov matrix $\vP$ has the following structure:
\begin{equation}
	\vP = \bar{\vP} + \vDelta,
\end{equation}
where $\bar{\vP}$ is a Markov matrix that is \textit{$r$-aggregatable}, i.e. there exists an $r$-cluster partition $\curlybrackets{\Omega_1, \Omega_2, \dots, \Omega_r}$ on the state space $[n]$ such that 
\begin{equation}
	\forall k\in[r], \forall i,j \in \Omega_k, \bar{\vP}(i,:)=\bar{\vP}(j,:).
\end{equation}
We assume $rank(\bar{\vP})=r$, which guarantees there are only $r$ unique rows in $\bar{\vP}$. Matrix $\vDelta$ is the perturbation that accounts for the difference of the true Markov matrix $\vP$ and the $r$-aggregatable Markov matrix $\bar{\vP}$. Note that so far we only assume modes are clustered based on their similarities in transition distributions and for future work we will take the mode dynamics and group connectivity into account.

\subsection{Problem Formulation}
Assuming parameters for all the modes $\curlybrackets{\vw_k}_{k=1}^n$ are known, given observation trajectory $\curlybrackets{y_t,u_t}_{t=0}^N$ with length $N$, we want to find an $r$-aggregatable approximation $\wtd{\vP}$ of $\vP$ such that the partition information in $\wtd{\vP}$ could recover $\curlybrackets{\Omega_1, \Omega_2, \dots, \Omega_r}$ in $\bar{\vP}$.

%Note that $\curlybrackets{\vw_k}_{k=1}^n$ can be estimated via system identification \cite{bemporad2018fitting, du2018robust}, so we assume they are known a priori here for simplicity.

We seek an $r$-aggregatable approximation of the original Markov matrix while preserving the clustering information in the underlying aggregatable Markov matrix. Given a Markov chain, one could use the power method \cite{Stewart1978} to iteratively simulate the evolution of the state distribution or compute the stationary distribution. So, one motivation to solve the aforementioned problem is that, during the power method, it requires $O(n^2)$ scalar multiplications in one iteration for $\vP$ but only $O(rn)$ for the $r$-aggregatable $\wtd{\vP}$. Meanwhile, the compromise in accuracy brought by the reduction of computation can be upper bounded with the following theorem.

\begin{theorem}\label{thrm_MCDiff}
	The differences between two Markov matrices $\vP$ and $\wtd{\vP}$ in terms of stationary distribution satisfy
	\begin{equation} \label{eq_40}
	\norm{\vpi - \tilde{\vpi}}_1 \leq \sum_{i=2}^{n} \frac{1}{1-\lambda_i(\vP)} \norm{\vP-\wtd{\vP}}_\infty.
	\end{equation}
	Furthermore, if $\vP$ and $\wtd{\vP}$ are both ergodic, their transient distributions and satisfy
	\begin{equation} \label{eq_41}
	\norm{\vpi_t - \tilde{\vpi}_t}_1 \leq C \rho^t + \norm{\vpi - \tilde{\vpi}}_1
	\end{equation}
	for some $C>0$ and $0<\rho<1$.
\end{theorem}

We can see that as long as the approximation error $\norm{\vP-\wtd{\vP}}_\infty$ is upper bounded, the stationary and transient behavior differences between the true Markov matrix $\vP$ and the $r$-aggregatable approximation $\wtd{\vP}$ can be bounded. This gives the justification for using $\wtd{\vP}$ as a surrogate for $\vP$ in the power method. The distance $\norm{\vP-\wtd{\vP}}_\infty$ with $\wtd{\vP}$ obtained from our approach is bounded in Theorem \ref{thrm_main_P}.

\section{Our Approach}\label{sec_approach}
Our approach to solve the problem mentioned above is given in Algorithm \ref{Alg_1}.
{\SetAlgoNoLine%
\begin{algorithm}[h!] 
	\KwIn{Observation $\curlybrackets{y_t,u_t}_{t=0}^N$, dynamics $\curlybrackets{\vw_k}_{k=1}^n$}
	\For{$t = 0, \dots, N$}{
		$\vphi_t := [y_{t-1}, \dots, y_{t-n_a}, u_{t-1}, \dots, u_{t-n_c}]^\T $\\
		$\wht{X}_t = \underset{k \in [n]}{\arg \min} |y_t - \vw_k^\T \vphi_t|$ \label{algline_1}\\
	}
	Compute empirical Markov matrix:
%	\begin{equation}
%		\wht{\vP}(i,j) = 
%		\begin{cases}
%		\frac{\sum_{t=1}^{N} \indicator \curlybrackets{\wht{X}_{t-1}=i, \wht{X}_t=j}}{\sum_{t=1}^{N-1} \indicator \curlybrackets{\wht{X}_t=i}} & \text{ if } \sum_{t=0}^{N-1} \indicator \curlybrackets{\wht{X}_t=i} \neq 0\\ 
%		1/n & \text{ o.w. } 
%		\end{cases}		
%	\end{equation}\\
	\begin{equation}
	\wht{\vP}(i,j) = \frac{\sum_{t=1}^{N} \indicator \curlybrackets{\wht{X}_{t-1}=i, \wht{X}_t=j}}{\sum_{t=1}^{N} \indicator \curlybrackets{\wht{X}_{t-1}=i}}
	\end{equation} \label{algline_2}\\
	SVD decomposition: $\wht{\vP} = \vU \vSigma \vV^\T$\\
	$\vU_r = \vU(:, 1{:}r)$ \label{algline_3}\\
	Solve the following k-means problem:
	\begin{equation} \label{eq_31}
	\hspace{0em}
	\wht{\Omega}_{1:r}, \hat{\vc}_{1:r} = 
	\underset{\substack{\wht{\Omega}_1, \dots, \wht{\Omega}_r \\ \hat{\vc}_1,\dots, \hat{\vc}_r }}{\arg \min}		
	\sum_{k = 1}^r \sum_{i \in \wht{\Omega}_k} 
	\norm{\vU_r(i,:) - \hat{\vc}_k}^2
	\end{equation} \label{algline_4} \\
	Aggregatable approximation: assume $i \in \wht{\Omega}_s$	
	\begin{equation}
		\wtd{\vP}(i,j) = \frac{\sum_{k\in\wht{\Omega}_s} \sum_{t=1}^{N} \indicator \curlybrackets{\wht{X}_{t-1}=k, \wht{X}_t=j}}{\sum_{k\in\wht{\Omega}_s} \sum_{t=1}^{N} \indicator \curlybrackets{\wht{X}_{t-1}=k}}
	\end{equation} \label{algline_5}\\
	\KwOut{Partition $\curlybrackets{\wht{\Omega}_1, \dots, \wht{\Omega}_r}$ and matrix $\wtd{\vP}$}
	\caption{Mode Clustering for Markov Jump Model} \label{Alg_1}
\end{algorithm}}
In Line \ref{algline_1}, we estimate the active mode at time $t$ by picking the mode whose dynamics gives the smallest residual error $|y_t - \vw_k^\T \vphi_t|$. Then, in Line \ref{algline_2}, based on the estimated mode sequence, we estimate $\vP$ with the empirical Markov matrix $\wht{\vP}$ in which the transition probability from mode $i$ to mode $j$ is estimated with the frequency of transition pair $(i,j)$ with respect to mode $i$. In Line \ref{algline_3}, we take the SVD of $\wht{\vP}$ and preserve the first $r$ singular value components. This is essentially a denoising step that reduces the influence of perturbation $\vDelta$ and estimation error in $\wht{\vP}$, and the obtained $\vU_r$ is a dimension-reduced representation of $\wht{\vP}$ that bears the low-rank structure in $\bar{\vP}$. Then, we use k-means to estimate the clustering information in $\bar{\vP}$. Finally, in Line \ref{algline_5}, we compute $\wtd{\vP}$ by taking modes within the same estimated cluster as a single mode and re-computing the empirical Markov matrix.

Note that if a certain mode does not show up at all in the trajectory, i.e. the denominators in Line \ref{algline_2} and Line \ref{algline_5} might be 0, then we simply assign uniform distribution to that mode, i.e. $\wht{\vP}(i,j) = 1/n$. We show in the proof that when the trajectory is long enough, every mode will show up with high probability.

\section{Theoretical Guarantees}\label{sec_theory}
\subsection{Relevant Definitions}
Before discussing theoretical guarantees of the proposed approach, we introduce some definitions that will be used later.

\begin{definition}[Mixing Time of MC] \label{def_mixingtime}
	Let $\vP \in \dm{n}{n}$ be a row stochastic Markov transition matrix with stationary distribution $\vpi$. Then for all $\epsilon > 0$, the $\epsilon-$mixing time is defined as
	\begin{equation}
	\tau(\epsilon) = \min \curlybracketsbig{k: \max_{i \in [n]} \frac{1}{2} \norm{(\vP^k)(i,:)^\T - \vpi }_1 \leq \epsilon }.
	\end{equation}
	Moreover, we let $\tau_* = \tau(\frac{1}{4})$.
\end{definition}

Since k-means is used in Algorithm \ref{Alg_1}, we assume a $(1+\epsilon)$ solution to the k-means problem can be obtained and later show how $\epsilon$ affects the overall clustering error.
\begin{definition}[Approximate Solution to k-means Clustering Problem]
	For problem in \eqref{eq_31}, we say $\wht{\Omega}_1, \dots, \wht{\Omega}_r, \hat{\vc}_1,\dots, \hat{\vc}_r$ is a $(1+\epsilon)$ solution if
	\begin{equation}
	\sum_{s = 1}^r \sum_{i \in \wht{\Omega}_s} \norm{\vU_r(i,:) - \hat{\vc}_s}^2
	\leq
	(1+\epsilon)
	\underset{\substack{\Omega_1, \dots, \Omega_r \\ \vc_1,\dots, \vc_r }}{\min}		
	\sum_{s = 1}^r \sum_{i \in \Omega_s} 
	\norm{\vU_r(i,:) - \vc_s}^2.
	\end{equation}
\end{definition}

\begin{definition}[Misclustering Rate]\label{def_MisclusteringRate}
	Let $\curlybrackets{\Omega_1, \Omega_2, \dots, \Omega_r}$ be the underlying true clustering partition of $[n]$ and $\curlybrackets{\wht{\Omega}_1, \wht{\Omega}_2, {\dots}, \wht{\Omega}_r}$ be an estimate of the true partition. We define \emph{misclustering rate} of $\curlybrackets{\wht{\Omega}_1, \wht{\Omega}_2, {\dots}, \wht{\Omega}_r}$ as
	\begin{equation} \label{eq_9}
	MR(\wht{\Omega}_1, \wht{\Omega}_2, \dots, \wht{\Omega}_r) = \min_{k \in \mathcal{K}} \sum_{j=1}^{r} \frac{|\curlybrackets{i: i \in \Omega_j; i \notin \wht{\Omega}_{k(j)}}|}{|\Omega_j|},
	\end{equation}
	where $\mathcal{K}$ is the set of all bijections from $[r]$ to $[r]$.
\end{definition}

Since the partition is invariant to the labels of clusters, when we evaluate the misclustering rate, we compute the error under the best label matching, which is the reason we need $\mathcal{K}$. Note that in \eqref{eq_9}, each summand has numerator no larger than the its denominator, so $M(\wht{\Omega}_1, \wht{\Omega}_2, \dots, \wht{\Omega}_r) \leq r$ trivially.

\subsection{Main Results}

Let $N' := \sum_{t=0}^{N-1} \indicator \curlybrackets{\wht{X}_t \neq X_t}$ denote the number of mistakes in the estimated mode sequence and $\eta:=\frac{N'}{N}$ denote the mistake rate. In the following analyses, Lemma \ref{lemma_NoMistakeCondition} gives conditions under which $N'=0$. Theorem \ref{thrm_main_MR} and Theorem \ref{thrm_main_P} give the upper bounds on misclustering rate and approximation error.

\begin{lemma}\label{lemma_NoMistakeCondition}
	Assume for all $t, |n_t|<n_{\max}$ and for all $ j \in [n]\backslash X_t$,
	\begin{equation} \label{eq_49}
		|\vphi_t^\T (\vw_{X_t} - \vw_j)| > 2 n_{\max} ,
	\end{equation}
	then the sequence estimated in Line \ref{algline_1} of Algorithm \ref{Alg_1} is correct, i.e. $N'=0$.
\end{lemma}

When $n_t=0$, the dynamics given in \eqref{eq_48} defines a hyperplane plus noise. Data points at the intersection of these hyperplanes (a set of measure zero in the noiseless case) are not useful in distinguishing the mode. \eqref{eq_49} essentially means that such data points do not exist.

%\begin{theorem} \label{thrm_main_MR}
%	Assume: (i) assumptions in Section \ref{subsec_1} hold; (ii) $\vP$ is ergodic; (iii) $\curlybrackets{\wht{\Omega}_1, \dots, \wht{\Omega}_r}$ is a  $(1+\epsilon_1)$ solution to the k-means problem; (iv) $	\norm{\vDelta} \leq \frac{\sigma_r(\bar{\vP})}{8 \sqrt{(2+\epsilon_1) r}} \sqrt{\frac{|\Omega_{(r)}|}{ |\Omega_{(1)}|}  + 1}$. Then, for any $\forall \epsilon_2 >0$ small enough, 
%	when $N \geq 200 \tau_* \pi_{\max} \log (\log ({\epsilon}_2^{-1})) \log ({\epsilon}_2^{-1}) {\epsilon}_2^{-2} $, with probability no less than
%	\begin{equation}
%	1 - \exp \parenthesesbig{- \frac{N}{200 \tau_* \pi_{\max} \log({\epsilon}_2^{-1}) {\epsilon}_2^{-2} }}
%	\end{equation}
%	we have
%	\begin{equation}\label{eq_30}
%	MR(\wht{\Omega}_1, \wht{\Omega}_2, \dots, \wht{\Omega}_r) \\
%	\leq
%	64(2+\epsilon_1) r \parenthesesbig{\frac{\norm{\vDelta}}{\sigma_r(\bar{\vP})} + 
%		\frac{4 (\epsilon_2+\eta) (\norm{\vDelta} + \norm{\bar{\vP}})}{\pi_{\min} \sigma_r(\bar{\vP})}}^2
%	\end{equation}
%	where (i) $\eta=0$ if the conditions in Lemma \ref{lemma_1} hold; $\eta=3N'/N$, if $\sum_{t=0}^{N} \indicator \curlybrackets{\wht{X}_t \neq X_t} = N'$.
%\end{theorem}

\begin{theorem} \label{thrm_main_MR}
	Assume: (i) the framework in Section \ref{subsec_1} holds; (ii) $\vP$ is ergodic; (iii) $\curlybrackets{\wht{\Omega}_1, \dots, \wht{\Omega}_r}$ is a  $(1+\epsilon_1)$ solution to the k-means problem; (iv) $	\norm{\vDelta} \leq \frac{\sigma_r(\bar{\vP})}{8 \sqrt{(2+\epsilon_1) r}} \sqrt{\frac{|\Omega_{(r)}|}{ |\Omega_{(1)}|}  + 1}$; (v) mistake rate $\eta<\frac{\pi_{\min}}{2}$. 
	Then for all $\epsilon_2 >0$, let $\tilde{\epsilon}_2 = \min \curlybracketsbig{\epsilon_2, \frac{\pi_{\min}}{2} - \eta, \frac{\pi_{\min}}{4 (\sigma_1(\bar{\vP}) + \norm{\vDelta})} \parenthesesbig{\frac{\sigma_r(\bar{\vP})}{8 \sqrt{(2+\epsilon_1) r}} \sqrt{\frac{|\Omega_{(r)}|}{ |\Omega_{(1)}|}  + 1} - \norm{\vDelta}} }$,
	if $N \geq 200 \tau_* \pi_{\max} \log (\tilde{\epsilon}_2^{-1}) \tilde{\epsilon}_2^{-2} \squarebrackets{\log(24 n \tau_*) + \log (\log (\tilde{\epsilon_2}^{-1}))}$, with probability no less than
	\begin{equation}
		1 - \exp \parenthesesbig{- \frac{N}{200 \tau_* \pi_{\max} \log(\tilde{\epsilon}_2^{-1}) \tilde{\epsilon}_2^{-2} }} ,
	\end{equation}
	we have
	\begin{equation}\label{eq_30}
		MR(\wht{\Omega}_1, \wht{\Omega}_2, \dots, \wht{\Omega}_r) \\
		\leq
		64(2+\epsilon_1) r \parenthesesbig{\frac{\norm{\vDelta}}{\sigma_r(\bar{\vP})} + 
		\frac{4 (\epsilon_2+1.5\eta) (\norm{\vDelta} + \norm{\bar{\vP}})}{\pi_{\min} \sigma_r(\bar{\vP})}}^2 .
	\end{equation}
\end{theorem}

In Theorem \ref{thrm_main_MR}, the ergodicity condition on Markov matrix $\vP$ and mistake rate $\eta \leq \frac{\pi_{\min}}{2}$ guarantees that $\vP$ can be well learned from a single trajectory. When $\epsilon_2$ is small enough, $\tilde{\epsilon}_2$ becomes $\epsilon_2$, which will be more interpretable for the probability and trajectory length lower bounds. Through some further inspection of Theorem \ref{thrm_main_MR}, we could see the bounds improve as any of the following decreases: number of modes $n$, number of clusters $r$, perturbation $\norm{\vDelta}$, mixing time $\tau_*$, condition number $\sigma_1(\bar{\vP}) / \sigma_r(\bar{\vP})$, and disparities in stationary distribution $\vpi$ and cluster population, namely $\pi_{\max} / \pi_{\min}$ and $|\Omega_{(1)}| / |\Omega_{(r)}|$. The disparities play a role here because as disparities increases, certain modes or clusters may be dominated by the others and become less likely to show up in the data. This will make them less learned in the algorithm and the estimation and clustering error will increase accordingly.

\begin{theorem}\label{thrm_main_P}
	Under the same conditions as Theorem \ref{thrm_main_MR}, if $MR=0$, then with the same probability lower bound we could have
	\begin{equation}
		\norm{\vP - \wtd{\vP}}_\infty \leq 12 \sqrt{n} \pi_{\min}^{-1} \sigma_1(\vP) (\epsilon_2 + 1.5 \eta) + 2 \norm{\vDelta}_\infty .
	\end{equation}
\end{theorem}

Theorem \ref{thrm_main_P} gives the upper bound on the approximation error of $\wtd{\vP}$, which can be used to upper bound the stationary and transient behavior differences in Theorem \ref{thrm_MCDiff}. The limitation of the theorem is that the result holds only when the clustering error is $0$.

\section{Experiments}\label{sec_experiment}
\subsection{Synthetic Data}
We first study the performance of our approach with synthetic data. In the Markov jump model, we let $n_a{=}3, n_c{=}2$ and number of modes $n=50$. For each mode, the dynamics are generated by uniformly sampling its poles on $(-1,1)$. We let input $u_t \sim \N(0,1)$ and noise $n_t \sim Unif(-n_{\max}, n_{\max})$. The state space $[n]$ is partitioned into $r$ clusters $\Omega_{1:r}$ randomly such that every possible partition is sampled with equal probability. The mode transition probabilities $\bar{\vP}(\Omega_k,:)$ for every $k$ and initial mode distribution $\pi_0$ are sampled from uniform Dirichlet distribution. 

The error metrics we evaluate are: (i) clustering error $\text{CE} = n^{-1} \min_{k \in \mathcal{K}} \sum_{j=1}^{r} |\curlybrackets{i: i \in \Omega_j; i \notin \wht{\Omega}_{k(j)}}|$ where $\mathcal{K}$ is given in Definition \ref{def_MisclusteringRate}; (ii) $\norm{\wtd{\vpi} - \vpi}_1$, i.e. the difference between $\wtd{\vP}$ and $\vP$ in terms of stationary distributions. For each parameter setup, we record the average of these two metrics over 100 experiments.

\subsubsection{Without Perturbation ($\vDelta=0$)} \label{subsubsec_1}
We first evaluate how the performance depend on number of clusters $r$ and noise magnitude $n_{\max}$. We set perturbation $\vDelta=0$ for these test cases. The experiment results are given in Fig.(\ref{subfig_1}-\ref{subfig_4}). We set $n_{\max}=0.1$ in Fig.(\ref{subfig_1}-\ref{subfig_2}) and $r=6$ in Fig.(\ref{subfig_3}-\ref{subfig_4}).

\subsubsection{With Perturbation ($\vDelta \neq 0$)}
In this test case, we fix $n=50, r=6, n_{\max}=0.05, N=10^5$. 
The space of $\vDelta$ is a polytope which makes it difficult to sample uniformly, so instead for $i \in \Omega_k$, we sample $\vP(i,:)$ from Dirichlet distribution with parameters $\alpha \vP(\Omega_k,:)$ and record $\vDelta = \vP - \bar{\vP}$. In this case, $\expctn[\vP(i,:)] = \vP(\Omega_k,:)$ and $\alpha$ controls how much $\vP(i,:)$ deviates from $\vP(\Omega_k,:)$. We sweep $\alpha$ and use scatter plots Fig.(\ref{subfig_5}-\ref{subfig_6}) to show how the error metrics vary with $\norm{\vDelta}$.

\begin{figure}[h!]
	\centering
	\captionsetup[subfloat]{captionskip=0pt, farskip=0pt}
	\subfloat[]{\includegraphics[width=2.5in]{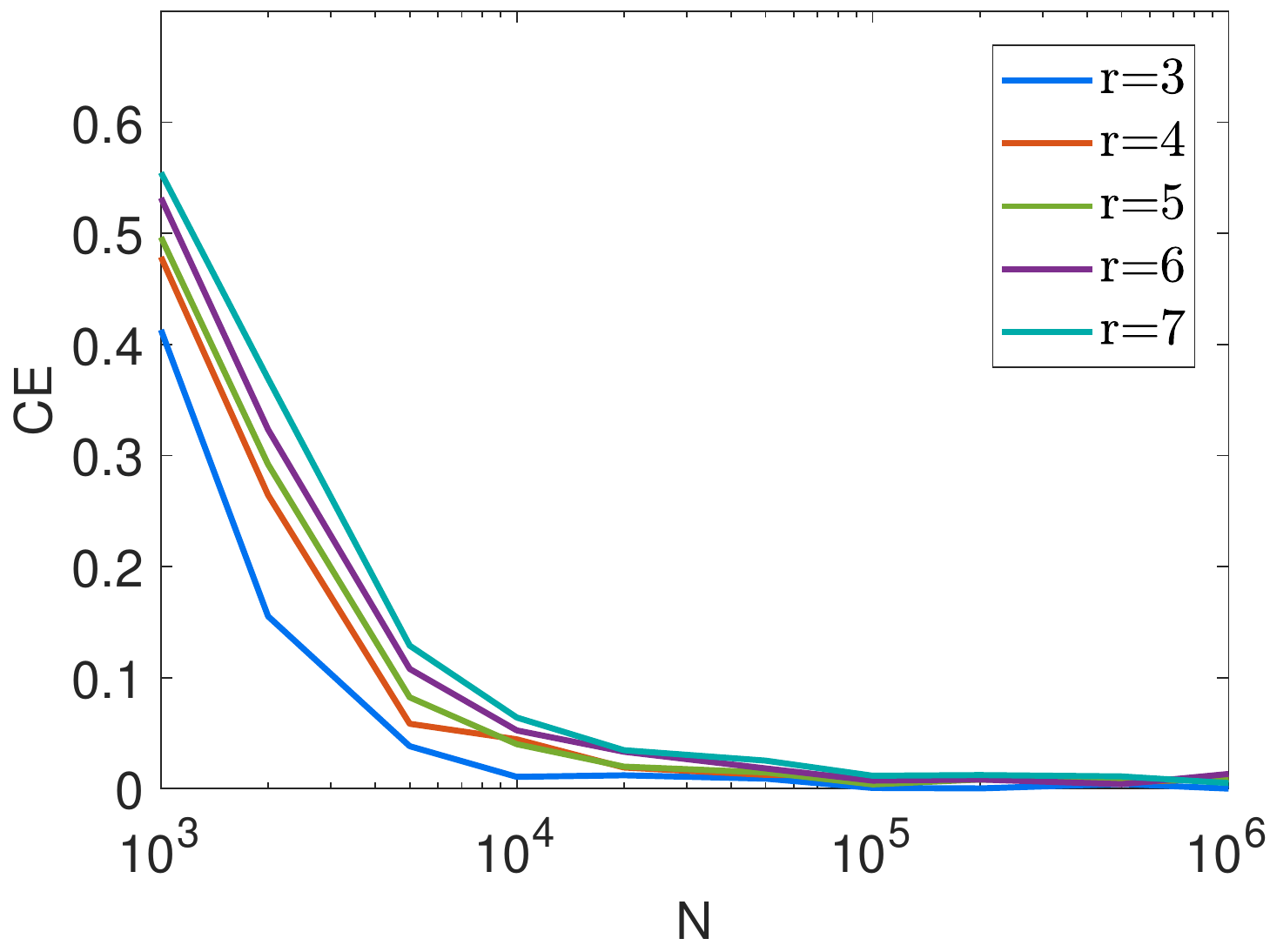} \label{subfig_1}} \ %\hfil
	\subfloat[]{\includegraphics[width=2.5in]{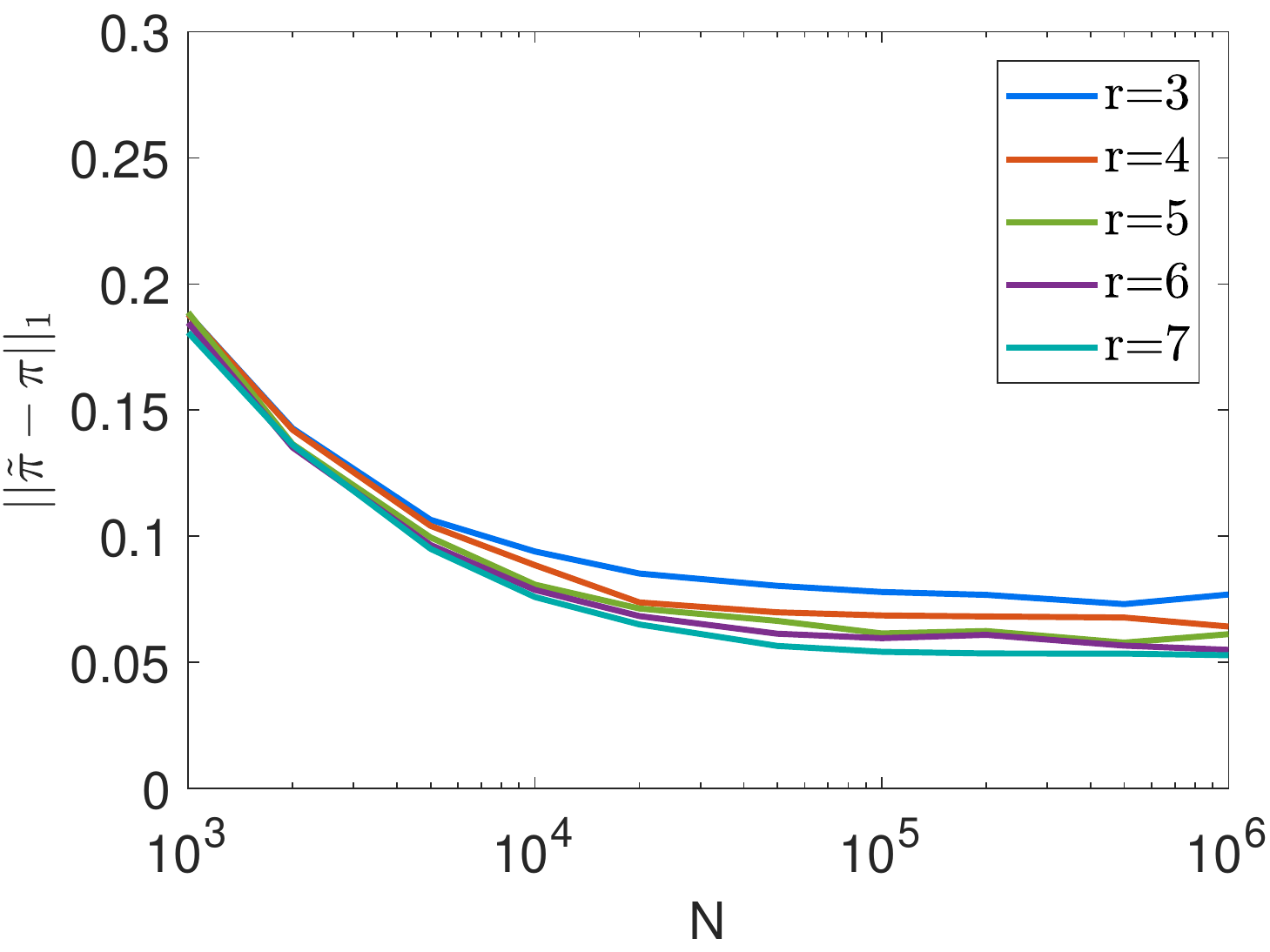} \label{subfig_2}} \\ 
	\subfloat[]{\includegraphics[width=2.5in]{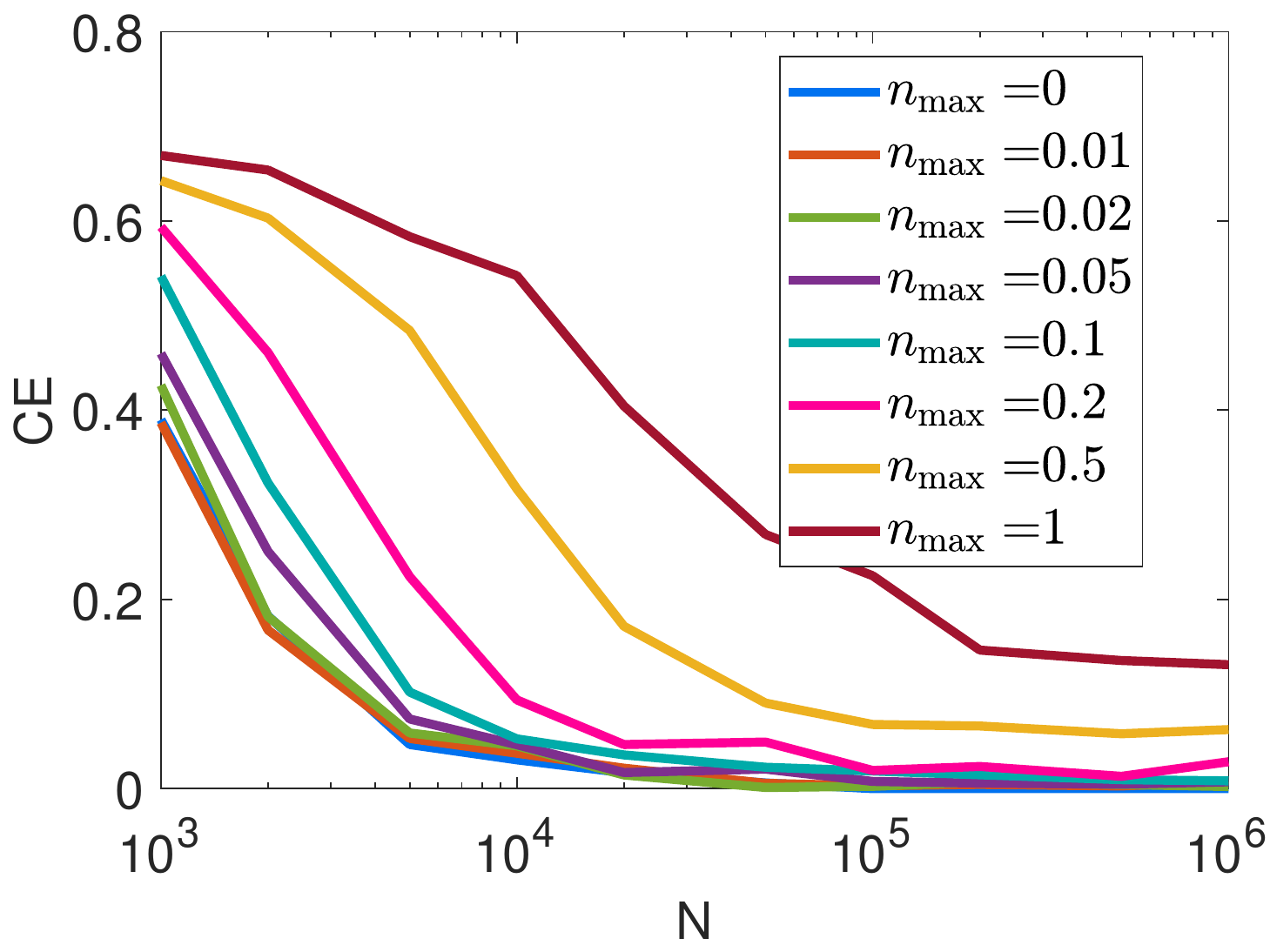} \label{subfig_3}} \ %\hfil
	\subfloat[]{\includegraphics[width=2.5in]{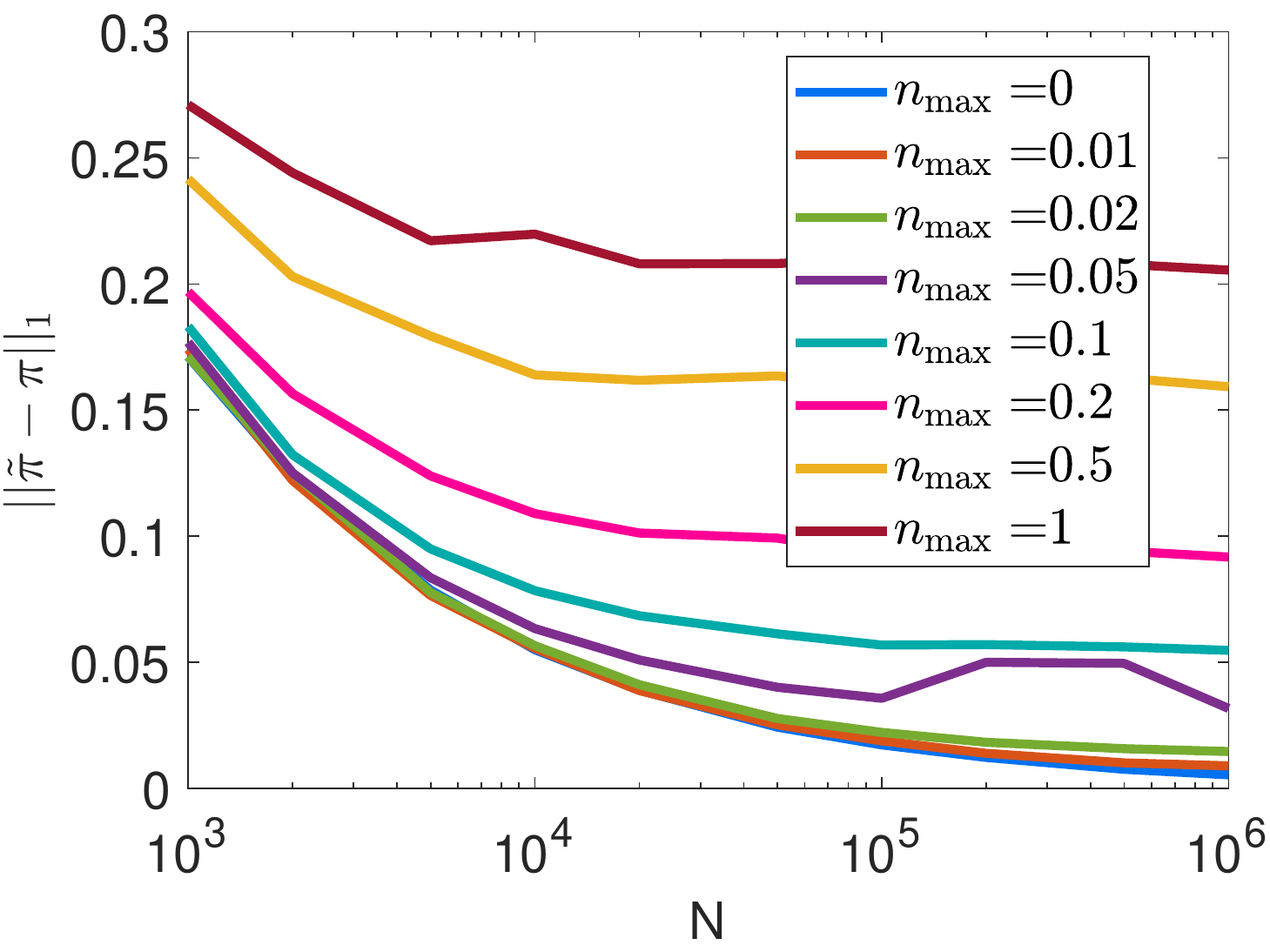} \label{subfig_4}} \\
	\subfloat[]{\includegraphics[width=2.5in]{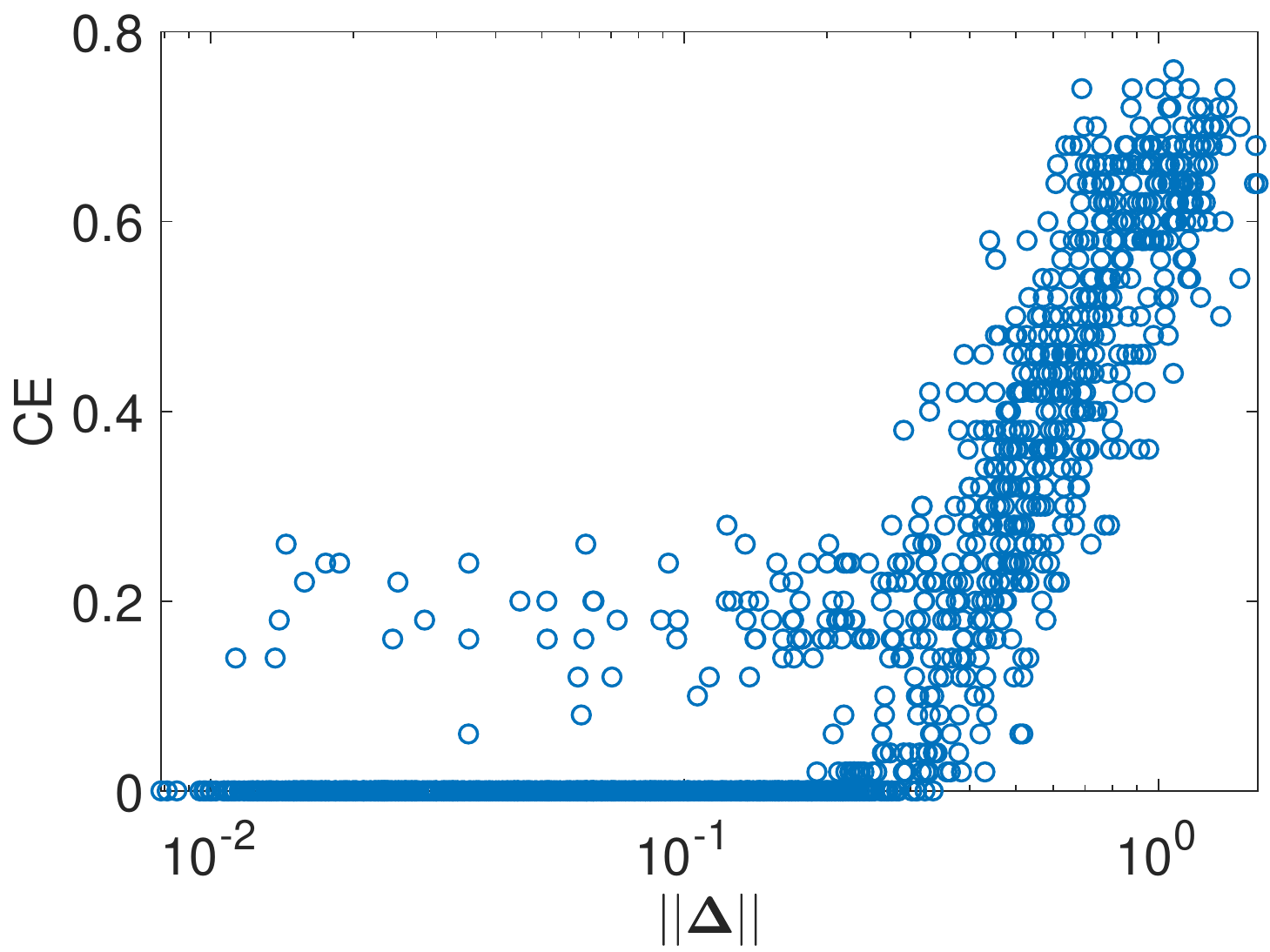} \label{subfig_5}} \
	\subfloat[]{\includegraphics[width=2.5in]{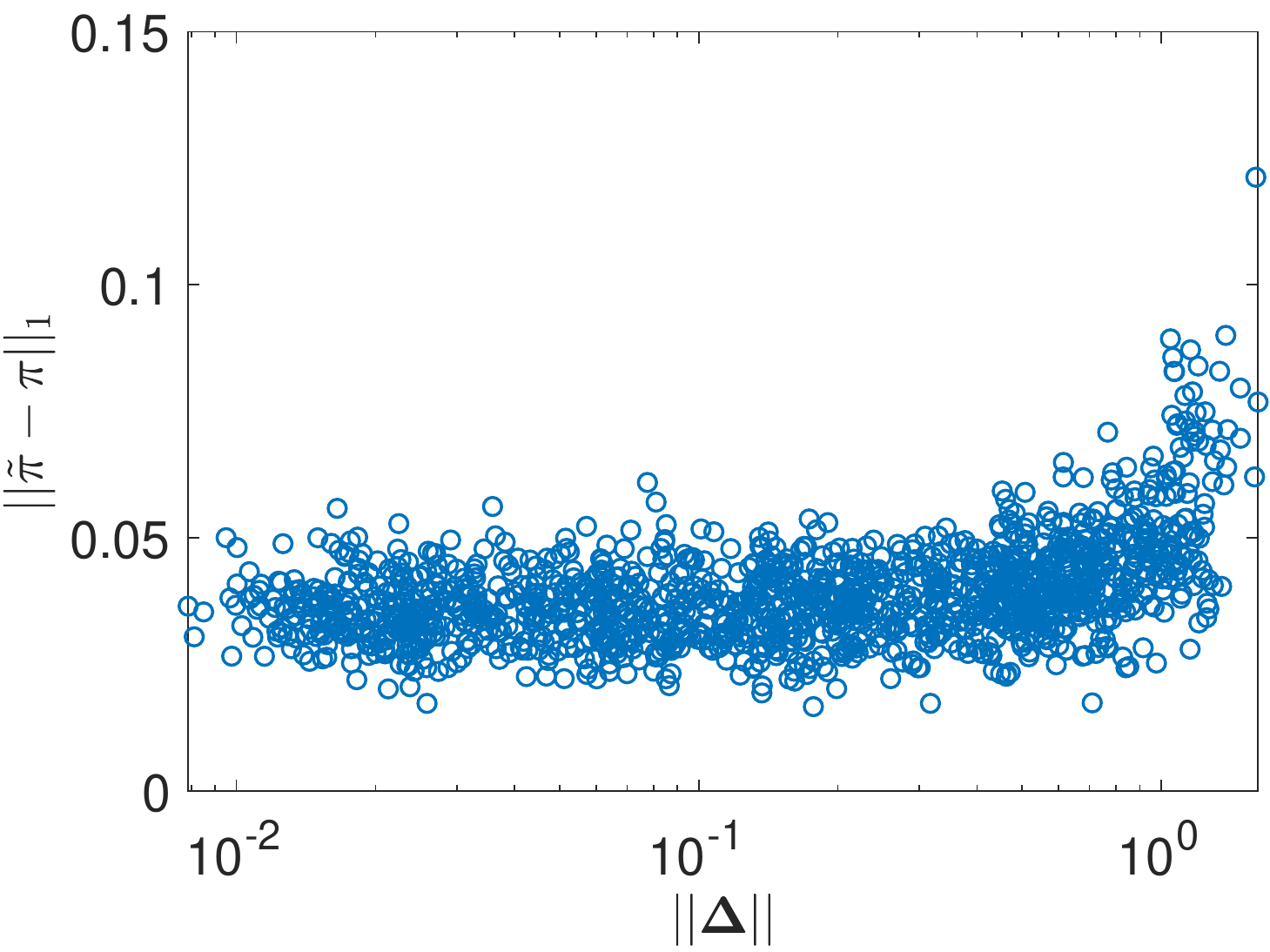} \label{subfig_6}} \
	\caption{\small Performance vs: (a,b) $N$ and $r$; (c,d) $N$ and $n_{\max}$; (e,f) $\norm{\vDelta}$}
	\label{fig_SynDataExp}
\end{figure}

\subsection{Practically Motivated Example---Patrol Robot}
Now we consider a more realistic case involving Markov jump system that can possibly benefit from our approach. Assume in a region, we have $n$ stations each with position $p_i \in \R$ and at time $t$ there is only one active station $s_t$ that generates requests; the sequence of active stations $s_{0:t}$ follows a Markov chain $\vP$. There is a robot with position $x_t \in \R$ at time $t$ aiming to reach the active station as fast and close as possible. Assuming the dynamics and control law of the robot are given by
\begin{equation}
\begin{split}
x_{t+1} &= x_t + u_t + n_t , \\
u_t &= K(p_{s_t} - x_t) ,
\end{split}
\end{equation}
the closed-loop dynamics take the form 
\begin{equation}
x_{t+1} = (1-K) x_t + K p_{s_t} + n_t ,
\end{equation}
which is a Markov jump model. In this setting, if the underlying Markov chain bears aggregatability property to some extent, we could use our approach to uncover the corresponding partition of modes as well as find an approximation of Markov transition matrix with stationary distribution that is easier to compute. Understanding the similarities between the stations' activation schedule can be useful to design improved control strategies for the robot. %And when the stationary distribution can be easily computed, we can better know which stations are more frequently to be active and design the scheduling if we have multiple robots to coordinate.

%Note that if the goal is to make $\sum_{t=0}^{N} |x_t - p_{s_t}|$ small, this control law has an obvious disadvantage: $x_{t+1}$ is driven to $s_t$ rather than $s_{t+1}$, which might make the error $|x_{t+1} - p_{s_{t+1}}|$  unnecessarily large. So, with our approach, a heuristically improved control law could be 
%\begin{equation} \label{eq_50}
%	u_t = K(p_{\arg \max_k \wtd{\vP}(s_t, k) } - x_t), 
%\end{equation}
%where it essentially predict what is the most likely future active station and set it as the destination for next time. The advantage of using $\wtd{\vP}$ which partitions the stations into $r$ clusters rather than the empirical transition matrix $\wht{\vP}$ is that the robots only need to store control laws for $r$ clusters. The advantage will become more prominent when we use more complicated control law than this proportional controller and when the number of stations becomes larger.

%In the experiment, we set $n=50, p_i = i, K=0.7, n_t \sim \N(0,0.1), N=10^6$ and sample $\vP$ from symmetric Dirichlet distribution with parameter 0.1. The average of  $N^{-1}\sum_{t=0}^{N} |x_t - p_{s_t}|$ over 100 runs is 14.10 for the original control law and 11.79 for the improved law. 
%
%If, in the Markov chain, for any mode $i$, $\max_j \vP(i,j)$ is even larger than the rest of $\vP(i,j)$'s, then the prediction of future state in \eqref{eq_50} would be more accurate and one could expect an even better improvement from the improved law.

In the experiment, we set $n=50, p_i = i, K=0.7, n_t \sim \N(0,0.1), N=10^6$ and sample $\bar{\vP}, \vP$ same as \ref{subsubsec_1}. Over the average of 100 runs, clustering error $\text{CE} = 0.04$ and $\norm{\wtd{\vpi} - \vpi}_1 = 0.07$.

\section{Conclusions \& Future Work}
In this paper, we consider the problem of model aggregation for Markov jump system from the perspective of clustering the modes based on their transition distributions. The proposed approach has guaranteed clustering error upper bound and exhibits decent performance in the experiments.

There are several interesting directions for future work: (i) we will see how lumpable Markov chain can help reformulate the model reduction problem; (ii) in the algorithm, after obtaining an estimate of the Markov transition matrix, one might use it to get a better estimate of the mode sequence, so several iterations between estimating switching sequence and Markov transition matrix may make both estimates more accurate; (iii) after the mode clustering, it is worth investigating if we could use a single mode to characterize the switching dynamics of all the modes within the cluster so that we could truly reduce the number of modes in the model.

\bibliographystyle{IEEEtran}
\bibliography{IEEEabrv,MarkovClustering}

\appendix

\section{Technical Lemmas} \label{sctn_lemmas}
This section provides technical lemmas that are used to prove the main results of the paper. The proofs for main theorems are given later in Section \ref{sctn_mainproof}.
%LEMMAS_ON_MATRIX_PERTURBATION_THEORY
\subsection{Lemmas on Matrix Perturbation Theory}
Note that in Line \ref{algline_3} of Algorithm \ref{Alg_1}, SVD is performed on matrix $\wht{\vP}$. To analyze this step, in this section, we give a few lemmas about how the perturbation of a matrix will affect its left singular vector matrix.

\begin{definition}[Distance between column spaces] \label{def_subspacedistance}
	For two matrices $\vE, \vF \in \dm{n}{n}, rank(\vE) {=} rank(\vF) {=} r$, let $\mathcal{E}, \mathcal{F}$ denote their column spaces and $\vPi_\vE$ and $\vPi_\vF$ denote corresponding projection matrices of the column spaces. Then, the principal angles \cite{knyazev2002principal} $\theta_1, \dots, \theta_r$ between $\mathcal{E}, \mathcal{F}$ can be shown to be the arcsines of first $r$ singular values of matrix $(\vI - \vPi_\vE) \vPi_\vF$. Let $\sin \vTheta(\vE, \vF):=\diag(\sin(\theta_1), \dots, \sin(\theta_r))$. The $\sin\vTheta$ distance between column space of $\vE$ and $\vF$ is defined as $\norm{\sin \vTheta(\vE, \vF)}_F$, which satisfies
	\begin{equation}
	\norm{\sin \vTheta (\vE, \vF) }_F = \norm{(\vI - \vPi_\vE) \vPi_\vF}_F.
	\end{equation}
\end{definition}

\begin{lemma}[Wedin's Perturbation Theorem \cite{wedin1972perturbation}] \label{lemma_Wedin}
	Let $\vA, \wht{\vA} \in \dm{m}{n}$. Let $\vA = \vU \vSigma \vV^\T$ be the SVD of $\vA$ where singular values on the diagonal of $\vSigma$ are arranged in descending order. Let $\vU_r = \vU(:,1{:}r),$ $\vSigma_r = \vSigma(1{:}r,1{:}r), \vV_r = \vV(:,1{:}r)$ denote the first $r$ singular components of matrix $\vA$. Similarly, let $\curlybrackets{\wht{\vU}_r, \wht{\vSigma}_r, \wht{\vV}_r}$ denote the first $r$ singular components of matrix $\wht{\vA}$. Then, if $\sigma_r(\wht{\vA}) - \sigma_{r+1}(\vA)>0$,
	\begin{equation}
	\max \curlybrackets{\norm{\sin \vTheta (\vV_r, \wht{\vV}_r)}, \norm{\sin \vTheta (\vU_r, \wht{\vU}_r)}} 
	\leq 
	\frac{\max \curlybrackets{\norm{(\vA-\wht{\vA}) \wht{\vV}_r}, \norm{\wht{\vU}_r^\T (\vA-\wht{\vA})}}}{\sigma_r(\wht{\vA}) - \sigma_{r+1}(\vA)} .
	\end{equation}
	Otherwise (note that in this case, we have $\sigma_r(\vA) - \sigma_{r+1}(\wht{\vA}) \geq 0$), we have
	\begin{equation}
	\max \curlybrackets{\norm{\sin \vTheta (\vV_r, \wht{\vV}_r)}, \norm{\sin \vTheta (\vU_r, \wht{\vU}_r)}} 
	\leq 
	\frac{\max \curlybrackets{\norm{(\vA-\wht{\vA}) {\vV_r}}, \norm{{\vU}_r^\T (\vA-\wht{\vA})}}}{\sigma_r({\vA}) - \sigma_{r+1}(\wht{\vA})} .
	\end{equation}	
	Note that the norm in the above equations can be replaced with any unitarily invariant norm.
\end{lemma}

\begin{lemma}[Weyl's Bound \cite{Weyl1912}] \label{lemma_Weyl}
	Given $\vA, \wht{\vA} \in \dm{m}{n}$, their singular values satisfy
	\begin{equation}
	\max_{i \leq \min \curlybrackets{m,n}} |\sigma_i(\vA) - \sigma_i(\wht{\vA})| \leq \norm{\vA - \wht{\vA}} .
	\end{equation}
\end{lemma}

\begin{lemma}[Combination of Wedin's Perturbation Theorem and Weyl's Bound] \label{lemma_AdvancedWedinBound}
	Let $\vA, \wht{\vA} \in \dm{m}{n}$. Let $\vA = \vU \vSigma \vV^\T$ be the SVD of $\vA$ where singular values on the diagonal of $\vSigma$ are arranged in descending order. Let $\vU_r = \vU(:,1{:}r), \vSigma_r = \vSigma(1{:}r,1{:}r), \vV_r = \vV(:,1{:}r)$ denote the first $r$ singular components of matrix $\vA$. Similarly, let $\curlybrackets{\wht{\vU}_r, \wht{\vSigma}_r, \wht{\vV}_r}$ denote the first $r$ singular components of matrix $\wht{\vA}$. Then
	\begin{equation}\label{eq_WedinModified}
	\max \curlybrackets{\norm{\sin \vTheta (\vV_r, \wht{\vV}_r)}, \norm{\sin \vTheta (\vU_r, \wht{\vU}_r)}}
	\leq 
	\frac{2 \norm{\vA - \wht{\vA}}}{\sigma_r(\vA) - \sigma_{r+1}(\vA)} .
	\end{equation}
\end{lemma}
\begin{proof}
	Note that if $2 \norm{\vA - \wht{\vA}} \geq \sigma_r(\vA) - \sigma_{r+1}(\vA)$, then \eqref{eq_WedinModified} holds trivially since $\norm{\sin \vTheta (\cdot, \cdot)} \leq 1$. So, it suffices to consider the case when $2 \norm{\vA - \wht{\vA}} < \sigma_r(\vA) - \sigma_{r+1}(\vA)$.
	
	Using Weyl's bound in Lemma \ref{lemma_Weyl}, we have
	\begin{gather}
	\sigma_r (\wht{\vA}) \geq \sigma_r (\vA) - \norm{\vA - \wht{\vA}} \label{eq_1} , \\
	\sigma_{r+1} (\wht{\vA}) \leq \sigma_{r+1} (\vA) + \norm{\vA - \wht{\vA}} \label{eq_2} .
	\end{gather}
	
	For the case when $\sigma_r(\wht{\vA}) - \sigma_{r+1}(\vA) >0 $, Wedin's perturbation theorem in Lemma \ref{lemma_Wedin} can give
	\begin{equation}
	\max \curlybrackets{\norm{\sin \vTheta (\vV_r, \wht{\vV}_r)}, \norm{\sin \vTheta (\vU_r, \wht{\vU}_r)}} 
	\leq 
	\frac{\norm{\vA - \wht{\vA}}}{\sigma_r(\wht{\vA}) - \sigma_{r+1}(\vA)} .
	\end{equation}
	Using \eqref{eq_1}, we have
	\begin{equation} \label{eq_3}
	\max \curlybrackets{\norm{\sin \vTheta (\vV_r, \wht{\vV}_r)}, \norm{\sin \vTheta (\vU_r, \wht{\vU}_r)}} 
	\leq
	\frac{\norm{\vA - \wht{\vA}}}{\sigma_r({\vA}) - \sigma_{r+1}(\vA) - \norm{\vA - \wht{\vA}}} .
	\end{equation}
	The condition $2 \norm{\vA - \wht{\vA}} < \sigma_r(\vA) - \sigma_{r+1}(\vA)$ implies that the RHS of \eqref{eq_3} is nonnegative and strictly smaller than 1, so adding $\norm{\vA - \wht{\vA}}$ to both the numerator and denominator of the RHS of \eqref{eq_3} shows \eqref{eq_WedinModified}.
	
	For the case when $\sigma_r(\wht{\vA}) - \sigma_{r+1}(\vA) \leq 0 $, the same result can be shown using Wedin's perturbation bound, equation \eqref{eq_2} and similar argument.	
\end{proof}

\begin{lemma} \label{lemma_subspacedistance}
	Let $\vU_1, \vU_2 \in \dm{n}{r}$ where $n \geq r$ and $\vU_1^\T \vU_1 = \vU_2^\T \vU_2 = \vI$. Then
	\begin{equation} \label{eq_5}
	\min_{\vO \in O(r)} \norm{\vU_1 - \vU_2 \vO}_F^2 \leq 2 \norm{\sin \vTheta (\vU_1, \vU_2)}_F^2 ,
	\end{equation}
	where $O(r)$ denotes the orthogonal group with dimension $r$, i.e. the set of all $r$-dimensional orthonormal matrices.
\end{lemma}
\begin{proof}
	Let the singular value decomposition of matrix $\vU_1^\T \vU_2$ be $\vU_1^\T \vU_2 = \vW \vSigma \vV^\T$. For any orthonormal $\vO$, we have
	\begin{equation}
	\begin{split}
	\norm{\vU_1 - \vU_2 \vO}_F^2 
	=& \tr \squarebrackets{(\vU_1 - \vU_2 \vO) (\vU_1^\T - \vO^\T \vU_2^\T)} \\
	=& \tr \squarebrackets{\vU_1 \vU_1^\T + \vU_2 \vU_2^\T - 2 \vO \vU_1^\T \vU_2} \\
	=& 2r - 2 \tr (\vO \vW \vSigma \vV^\T) \\
	=& 2r - 2 \tr (\vSigma \wtd{\vO}) \qquad (\wtd{\vO}:=\vV^\T \vO \vW) \\
	=& 2r - 2 \sum_{i=1}^{r} \vSigma(i,i) \wtd{\vO}(i,i) .
	\end{split}
	\end{equation}
	Following the definition of $\wtd{\vO}$, we see $\wtd{\vO}$ is orthonormal, thus $\wtd{\vO}(i,i) \leq 1,  \forall i \in [r]$. Also note that $\vSigma(i,i) \geq 0$, we have $\sum_{i=1}^{r} \vSigma(i,i) \wtd{\vO}(i,i) \leq \sum_{i}^{r} \vSigma(i,i) = \tr(\vSigma)$, where equality holds when $\wtd{\vO}(i,i)=1, \forall i \in [r]$, i.e. $\wtd{\vO} = \vI$ and $\vO = \vV \vW^\T$. Therefore,
	\begin{equation} \label{eq_10}
	\min_{\vO \in O(r)} \norm{\vU_1 - \vU_2 \vO}_F^2 \leq 2 (r - \tr (\vSigma)) .
	\end{equation}
	
	For the RHS of \eqref{eq_5}, we have
	\begin{equation} \label{eq_11}
	\begin{split}
	\norm{\sin \vTheta (\vU_1, \vU_2)}_F^2 
	= & \norm{(\vI - \vU_1 \vU_1^\T) \vU_2 \vU_2^\T}_F^2 \\
	= & \tr \squarebrackets{\vU_2 \vU_2^\T (\vI - \vU_1 \vU_1^\T) (\vI - \vU_1 \vU_1^\T) \vU_2 \vU_2^\T} \\
	= & \tr \squarebrackets{ (\vI - \vU_1 \vU_1^\T) \vU_2 \vU_2^\T} \\
	= & \tr (\vU_2 \vU_2^\T) - \tr (\vU_1^\T \vU_2 \vU_2^\T \vU_1) \\
	= & r - \tr (\vW \vSigma \vV^\T \vV \vSigma \vW^\T) \\
	= & r - \tr (\vSigma^2) , \\	
	\end{split}
	\end{equation}
	where the first equality follows Definition \ref{def_subspacedistance}. Since $\vSigma(i,i) \leq 1, \forall i\in [r]$, we have $\tr(\vSigma) \geq \tr (\vSigma^2)$. Following this and together with \eqref{eq_10} and \eqref{eq_11}, we can conclude the proof.
\end{proof}

%LEMMAS_ON_CLUSTERING
\subsection{Lemmas on Clustering}
In this section, we present some lemmas that can be applied to the clustering analysis in our algorithm.
\begin{definition}[Membership Matrix and Membership Matrix Set]
	We call a matrix $\vA \in \dm{n}{r}$ where $r\leq n$ a \emph{membership matrix} for $n$ points and $r$ clusters if each row of $\vA$ has exactly one element equal to $1$ and $0$'s for the rest of elements. And $A(i,j) = 1$ if and only if ``point $i$ belongs to cluster $j$''.	
	
	We let $\mathbb{M}_{n,r}:= \curlybrackets{\vM \mid \vM \in \dm{n}{r}, \vM \text{ is a membership matrix}}$ denote the set of all membership matrices for $n$ points and $r$ clusters.
\end{definition}
\begin{remark}
	Note that membership matrix is permutation invariant in the sense that, for any permutation matrix $\vQ \in \dm{r}{r}$, membership matrix $\vA$ and $\vA \vQ$ represent the same membership information and the only difference is that $\vQ$ changes the cluster labels. So, the permutation invariance establishes an equivalence relation among the set $\mathbb{M}_{n,r}$.
\end{remark}

\begin{lemma} \label{lemma_3}
	Let $\vX \in \dm{n}{m}$ be an arbitrary matrix with factorization $\vX = \vM_\vX \vC_\vX$, where $\vM_\vX \in \mathbb{M}_{n,r}$ is a membership matrix and $\vC_\vX \in \dm{r}{m}, rank(\vC_\vX)=r$. Let $\Omega_i = \curlybrackets{j \mid \vM_\vX(j,i)=1}, \forall i \in [r]$ and $\vU_r \vSigma_r \vV_r^\T$ be the SVD of $\vX$, which preserves only the first $r$ singular value components. Then, for any $ i \in \Omega_k$ and any $j \in \Omega_l$,
	\begin{equation}
	\norm{\vU_r(i,:) - \vU_r(j,:)} 
	= 
	\begin{cases}
	0 & \text{ if } k=l \\ 
	\sqrt{\frac{1}{|\Omega_k|} + \frac{1}{|\Omega_l|}}& \text{ if } k \neq l 
	\end{cases}.
	\end{equation}
\end{lemma}
\begin{proof}
	Since $\vX = \vM_\vX \vC_\vX$ and $\vC_\vX$ has full row rank, we see $rank(\vX)=r$, thus $\vSigma_r \vV_r^\T$ has full row rank as well. From this, we have (i) $\forall k, \forall i,j \in \Omega_k, \vX(i,:) = \vX(j,:)$, so $\vU_r(i,:)=\vU_r(j,:)$; (ii) $\forall k\neq l, \forall i \in \Omega_k, \forall j \in \Omega_l, \vX(i,:) \neq \vX(j,:)$, so $\vU_r(i,:) \neq \vU_r(j,:)$. Therefore, $\vU_r$ has factorization $\vU_r = \vM_\vX \vC$ for some $\vC \in \dm{r}{r}$. In another way, there are only $r$ unique rows in $\vU_r$ with $\vC$ collecting the unique rows and $\vM_\vX$ being the membership matrix shared with $\vX$.
	
	Since $\vI = \vU_r^\T \vU_r = \vC^\vT \vM_\vX^\T \vM_\vX \vC = \vC^\T \diag ([|\Omega_1|, |\Omega_{2}|, \dots, |\Omega_r|]) \vC$, we can see that matrix $\vC$ has orthogonal rows and $\norm{\vC(i,:)} = \frac{1}{\sqrt{|\Omega_i|}}, \forall i \in [r]$. Therefore, $\forall i \in \Omega_k, j \in \Omega_l$,
	\begin{equation}
	\norm{\vU_r(i,:) - \vU_r(j,:)} = \norm{\vC(k,:) - \vC(l,:)}
	= 
	\begin{cases}
	0 & \text{ if } k=l \\ 
	\sqrt{\frac{1}{|\Omega_k|} + \frac{1}{|\Omega_l|}}& \text{ if } k \neq l 
	\end{cases} .
	\end{equation}
\end{proof}

\begin{lemma}[Approximate k-means error bound, Lemma 5.3 in \cite{lei2015consistency}] \label{lemma_kmeans} 
	
	For $\epsilon>0$ and any two matrices $\vU, \bar{\vU} \in \dm{n}{r}$ such that $\bar{\vU}=\bar{\vM} \bar{\vC}$ with $\bar{\vM} \in \mathbb{M}_{n,r}, \bar{\vC} \in \dm{r}{r}$, let $\curlybrackets{{\vM}, {\vC}}$ be a $(1+\epsilon)$ solution to the k-means problem on ${\vU}$:
	\begin{equation*}
	\begin{array}{ll}
	&  {\vM} \in \mathbb{M}_{n,r}, {\vC}\in \dm{r}{r} \\ 
	\textbf{s.t.} 	& \norm{{\vM}{\vC} -  {\vU} }_F^2 \leq (1+\epsilon) \underset{\vM' \in \mathbb{M}_{n,r}, {\vC}'\in \dm{r}{r}}{\min} \norm{{\vM}'{\vC}' -  {\vU} }_F^2\\
	\end{array}.
	\end{equation*}
	Let $\Omega_k=\curlybrackets{i \mid i \in [n], \bar{\vM}(i,k)=1}$ denote the set of all points belonging to cluster k.	For any 
	\begin{equation} \label{eq_15}
	\delta_k \leq \min_{l \neq k} \norm{\bar{\vC}(l,:) - \bar{\vC}(k,:)}, \qquad \forall k\in [r], 
	\end{equation}
	define the set $S_k = \curlybrackets{i \mid i \in \Omega_k, \norm{({\vM}{\vC})(i,:) - \bar{\vU}(i,:)} \geq \delta_k/2}$, then
	\begin{equation} \label{eq_13}
	\sum_{k=1}^{r} |S_k| \delta_k^2 \leq 4(4+2\epsilon) \norm{\bar{\vU} - {\vU}}_F^2 .
	\end{equation}
	Moreover, define $G = \bigcup_{k=1}^r (\Omega_k\backslash S_k)$. If
	\begin{equation} \label{eq_12}
	(16+8\epsilon) \norm{\bar{\vU} - {\vU}}_F^2 / \delta_k^2 < |\Omega_k|, \qquad \forall k \in [r], 
	\end{equation}
	then there exists an $r\times r$ permutation matrix $\vJ$ such that $\bar{\vM} ({G,:}) = \vM({G,:}) \vJ $, i.e. $\vM$ and $\bar{\vM}$ share the same membership information for points in the set G.
\end{lemma}
\begin{remark}
	Lemma \ref{lemma_kmeans} can be used to bound the number of mis-clustered points by k-means. In this lemma, $\vU$ represents the data matrix (possibly dimension reduced) for $n$ data points. One applies k-means to $\vU$ and obtains the membership matrix $\vM$ and cluster centers $\vC$. $\bar{\vU}$ represents the ``clean'' data such that data points within the same true cluster have exactly the same rows and its membership information $\bar{\vM}$ is what one wants to recover and compares with. And its main takeaway is, under \eqref{eq_13}, at least the points in set $G$ can be clustered correctly from any $(1+\epsilon)$ solution of the k-means problem.
\end{remark}

%LEMMAS_ON_MARKOV_CHAIN_CONCENTRATION
\subsection{Lemmas on Markov Chain Concentration}
In this section, Lemma \ref{lemma_ConcentrationofMC} provides the estimation error bounds for Markov matrix estimation from a single trajectory generated by the Markov chain. Lemma \ref{lemma_ConcentrationofMC_withMistake} analyzes the case when certain states in the trajectory are perturbed. Lemma \ref{lemma_MixingTimeRelation} and Lemma \ref{lemma_VanillaConcentrationofMC} are building blocks towards Lemma \ref{lemma_ConcentrationofMC} and Lemma \ref{lemma_ConcentrationofMC_withMistake}.

\begin{lemma}[Lemma 5 in \cite{zhang2018state}] \label{lemma_MixingTimeRelation}
	Let $\tau(\epsilon)$ be the mixing time of Markov chain given in Definition \ref{def_mixingtime}. For any $\epsilon\leq \delta < 1/2$, we have
	\begin{equation}
	\tau(\epsilon) \leq \tau(\delta) \parenthesesbig{\left \lceil \frac{\log(\epsilon/\delta)}{\log(2\delta)} \right \rceil+1} .
	\end{equation}
\end{lemma}

\begin{lemma}[Lemma 7 in \cite{zhang2018state}] \label{lemma_VanillaConcentrationofMC}
	Let $\vP\in\dm{n}{n}$ be an ergodic row stochastic matrix with stationary distribution $\vpi\in\dm{n}{1}$. Let $\pi_{\max} = \max_i \vpi_i, \pi_{\min} = \min_i \vpi_i$. Let $\tau(\cdot)$ denote the mixing time of $\vP$, which is given in Definition \ref{def_mixingtime}. Let $\vF = \diag(\vpi) \vP$ denote the stationary frequency matrix. Given a trajectory $X_{0:N}$ of the Markov chain, define $\wht{\vpi}_0\in \dm{1}{n}$ and $\wtd{\vF} \in \dm{n}{n}$ as
	\begin{gather}
	\wht{\vpi}_0(i) = \frac{1}{N} \sum_{t=1}^{N} \indicator \curlybrackets{X_{t-1}=i} , \\
	\wht{\vF}_0(i,j) = \frac{1}{N} \sum_{t=1}^{N} \indicator \curlybrackets{X_{t-1}=i, X_t=j} .
	\end{gather}
	For any $ \epsilon>0$, let $\alpha = \tau(\min(\epsilon/2, \pi_{\max}))+1$, then
	\begin{gather}
	\prob \parenthesesbig{\norm{\wht{\vF}_0 - \vF} \geq \epsilon} \leq 2 \alpha n \exp \parenthesesbig{- \frac{N \epsilon^2 /8}{2 \pi_{\max} \alpha + \epsilon \alpha / 6}} \label{eq_18} , \\
	\prob \parenthesesbig{\norm{\wht{\vpi}_0 - \vpi}_\infty \geq \epsilon} \leq 2 \alpha n \exp \parenthesesbig{- \frac{N \epsilon^2 /8}{2 \pi_{\max} \alpha + \epsilon \alpha / 6}} \label{eq_19}  .
	\end{gather}
\end{lemma}

\begin{lemma} \label{lemma_ConcentrationofMC}
	Consider the Markov chain given in Lemma \ref{lemma_VanillaConcentrationofMC} and its trajectory $X_{0:N}$. Define $\tau_* = \tau(1/4)$ and $\wht{\vP}_0 \in \dm{n}{n}$ as
	\begin{equation} \label{eq_22}
	\wht{\vP}_0(i,j) = 
	\begin{cases}
	\frac{\sum_{t=1}^{N} \indicator \curlybrackets{X_{t-1}=i, X_t=j}}{\sum_{t=1}^{N} \indicator \curlybrackets{X_{t-1}=i}} & \text{ if } \sum_{t=1}^{N} \indicator \curlybrackets{X_{t-1}=i} \neq 0\\ 
	1/n & \text{ o.w. } 
	\end{cases} ,
	\end{equation}
	For any $ \epsilon>0$, let $\tilde{\epsilon} = \min \curlybrackets{\pi_{\min}/2, \epsilon}$, then we have
	\begin{equation}
	\prob \parenthesesbig{\norm{\wht{\vP}_0 - \vP} \leq 4 \pi_{\min}^{-1} \norm{\vP} \epsilon} 
	\geq  
	1-24 n \tau_* \log(\tilde{\epsilon}^{-1}) 
	\exp \parenthesesbig{- \frac{N}{100 \tau_* \pi_{\max} \log(\tilde{\epsilon}^{-1}) \tilde{\epsilon}^{-2} }} .
	\end{equation}
\end{lemma}
\begin{proof}
	For simplicity, restrict $\epsilon<\pi_{\min}/2$ for now. From Lemma \ref{lemma_VanillaConcentrationofMC}, we have the concentration results of $\wht{\vF}_0, \wht{\vpi}_0$ given in \eqref{eq_18} and \eqref{eq_19}, we will first simplify them before applying them to this proof.
	
	Since $\alpha = \tau(\min(\epsilon/2, \pi_{\max}))+1$ in Lemma \ref{lemma_VanillaConcentrationofMC} and we restrict $\epsilon<\pi_{\min}/2$ for now, we see $\alpha = \tau(\epsilon/2)+1$. We can obtain an upper bound on $\alpha$ as follows:
	\begin{equation} \label{eq_21}
	\begin{split}
	\alpha 
	&= \tau(\epsilon/2)+1 \\
	&\overset{\text{(i)}}{\leq}  \parenthesesbig{\left \lceil \frac{\log(2\epsilon)}{\log(1/2)} \right \rceil+1} \tau_* + 1 \\
	&\overset{\text{(ii)}}{\leq}  \parenthesesbig{\left ( \frac{\log(2\epsilon)}{\log(1/2)} \right )+3} \tau_*\\
	&= \parenthesesbig{\log_2 (0.5 \epsilon^{-1})+3} \tau_* \\	
	&\leq \parenthesesbig{\log_2 (\epsilon^{-1})+3} \tau_* \\	
	&\overset{\text{(iii)}}{\leq} 4\tau_* \log_2 (\epsilon^{-1}) \\
	&\leq 6\tau_* \log (\epsilon^{-1}) ,\\	
	\end{split}
	\end{equation}
	where: (i) holds since $\epsilon/2<\pi_{\min}/4\leq1/(4n)\leq1/4$ so we could apply Lemma \ref{lemma_MixingTimeRelation}; (ii) holds since mixing time $\tau_*\geq 1$; (iii) holds since $\epsilon^{-1} > 2 \pi_{\min}^{-1} \geq 2n \geq 2$ thus $\log_2(\epsilon^{-1})\geq 1$. Plugging \eqref{eq_21} into the RHS of \eqref{eq_18} and \eqref{eq_19}, we have
	\begin{equation} \label{eq_23}
	\begin{split}
	&2 \alpha n \exp \parenthesesbig{- \frac{N \epsilon^2 /8}{2 \pi_{\max} \alpha + \epsilon \alpha / 6}} \\
	\leq &12 n \tau_* \log(\epsilon^{-1}) \exp \parenthesesbig{- \frac{N}{8 \tau_* \log (\epsilon^{-1}) \epsilon^{-1} (12 \pi_{\max} \epsilon^{-1} + 1) }} \\
	\leq &12 n \tau_* \log(\epsilon^{-1}) \exp \parenthesesbig{- \frac{N}{100 \tau_* \pi_{\max} \log (\epsilon^{-1}) \epsilon^{-2}}} ,
	\end{split}
	\end{equation}
	where the last line holds since $0.5 \pi_{\max} \epsilon^{-1} > \pi_{\max} / \pi_{\min} \geq 1$. 
	
	In the remainder of the proof, we require the conditions $\norm{\wht{\vF}_0 - \vF} \leq \epsilon, \norm{\wht{\vpi}_0 - \vpi}_\infty \leq \epsilon$ to be satisfied. By applying union bound to \eqref{eq_18} and \eqref{eq_19} and plugging in \eqref{eq_23}, we can see
	\begin{equation} \label{eq_25}
	\prob \parenthesesbig{\norm{\wht{\vF}_0 - \vF} \leq \epsilon, \norm{\wht{\vpi}_0 - \vpi}_\infty \leq \epsilon} \geq 1 - 24 n \tau_* \log(\epsilon^{-1}) \exp \parenthesesbig{- \frac{N}{100 \tau_* \pi_{\max} \log (\epsilon^{-1}) \epsilon^{-2}}} .
	\end{equation}
	
	When $\norm{\wht{\vpi}_0 - \vpi}_\infty \leq \epsilon < \pi_{\min}/2 $, it is easy to see $\min_i \wht{\vpi}_0(i) \geq \pi_{\min}/2>0$, which implies every state has showed up at least once in the trajectory since otherwise there would be $0$ element in $\wht{\vpi}_0$. Also, by the definition of $\wht{\vF}_0, \wht{\vpi}_0, \wht{\vP}_0$, we can see $\wht{\vP}_0$ is determined only by the first line of \eqref{eq_22} thus $\wht{\vP}_0 = \diag (\wht{\vpi}_0)^{-1} \wht{\vF}_0$ holds. Now, consider $\norm{\wht{\vP}_0 - \vP}$, we have
	\begin{equation} \label{eq_24}
	\begin{split}
	\norm{\wht{\vP}_0 - \vP}
	&= \norm{\diag(\wht{\vpi}_0)^{-1} \wht{\vF}_0 - \diag(\vpi)^{-1} \vF} \\
	&\leq \norm{\diag(\wht{\vpi}_0)^{-1} (\wht{\vF}_0 - \vF)} + \norm{\parenthesesbig{\diag(\wht{\vpi}_0)^{-1} - \diag (\vpi)^{-1}} \vF}\\
	&\leq \norm{\diag(\wht{\vpi}_0)^{-1}} \norm{(\wht{\vF}_0 - \vF)} + \norm{\vI - \diag (\vpi / \wht{\vpi}_0)} \norm{\diag(\vpi)^{-1} \vF} \\
	&= \parenthesesbig{\underset{i}{\min} \wht{\vpi}_0(i)}^{-1} \norm{(\wht{\vF}_0 - \vF)} + \underset{i}{\max} \frac{|\wht{\vpi}_0(i) - \vpi_i|}{\wht{\vpi}_0(i)} \norm{\vP} \\
	&\leq 2 \pi_{\min}^{-1} \norm{(\wht{\vF}_0 - \vF)} + \frac{\max_i |\wht{\vpi}_0(i) - \vpi_i|}{\min_j \wht{\vpi}_0(j)} \norm{\vP} \\
	&\overset{\text{(i)}}{\leq} 2 \pi_{\min}^{-1} \epsilon + 2 \pi_{\min}^{-1} \epsilon \norm{\vP} \\
	&\overset{\text{(ii)}}{\leq} 4 \pi_{\min}^{-1} \norm{\vP} \epsilon ,
	\end{split}
	\end{equation}
	where: (i) holds as $\norm{\wht{\vF}_0 - \vF} \leq \epsilon, \norm{\wht{\vpi}_0 - \vpi}_\infty \leq \epsilon$ and (ii) holds as $\norm{\vP} \geq \norm{\vP \frac{1}{\sqrt{n} \onevec }} = \norm{\frac{1}{\sqrt{n} \onevec }} = 1$. Therefore, for any $\epsilon \leq \pi_{\min}/2$, by combining \eqref{eq_24} and \eqref{eq_25}, we have
	\begin{equation}
	\prob \parenthesesbig{\norm{\wht{\vP}_0 - \vP} \leq 4 \pi_{\min}^{-1} \norm{\vP} \epsilon} 
	\geq  
	1-24 n \tau_* \log({\epsilon}^{-1}) 
	\exp \parenthesesbig{- \frac{N}{100 \tau_* \pi_{\max} \log({\epsilon}^{-1}) {\epsilon}^{-2} }} .
	\end{equation}
	Finally, we could remove the restriction $\epsilon \leq \pi_{\min}/2$. We have, for any $\epsilon>0$ and let $\tilde{\epsilon} = \min \curlybrackets{\epsilon, \pi_{\min}/2}$, then
	\begin{equation}
	\begin{split}
	\prob \parenthesesbig{\norm{\wht{\vP}_0 - \vP} \leq 4 \pi_{\min}^{-1} \norm{\vP} \epsilon} 
	&\geq
	\prob \parenthesesbig{\norm{\wht{\vP}_0 - \vP} \leq 4 \pi_{\min}^{-1} \norm{\vP} \tilde{\epsilon}} \\
	&\geq 	
	1-24 n \tau_* \log({\tilde{\epsilon}}^{-1}) 
	\exp \parenthesesbig{- \frac{N}{100 \tau_* \pi_{\max} \log(\tilde{\epsilon}^{-1}) \tilde{\epsilon}^{-2} }} ,
	\end{split}
	\end{equation}	
	which concludes the proof.
\end{proof}

\begin{lemma} \label{lemma_ConcentrationofMC_withMistake}
	Consider all the conditions given in Lemma \ref{lemma_VanillaConcentrationofMC} and \ref{lemma_ConcentrationofMC} except that there are $N'$ perturbations in the trajectory of Markov chain, i.e. $\sum_{t=0}^{N} \indicator \curlybrackets{\wht{X}_t \neq X_t} = N'$, where $\wht{X}_{0:N}$ denotes the perturbed trajectory. Let
	\begin{equation}
	\wht{\vP}(i,j) = 
	\begin{cases}
	\frac{\sum_{t=1}^{N} \indicator \curlybrackets{\wht{X}_{t-1}=i, \wht{X}_t=j}}{\sum_{t=1}^{N} \indicator \curlybrackets{\wht{X}_{t-1}=i}} & \text{ if } \sum_{t=1}^{N} \indicator \curlybrackets{\wht{X}_{t-1}=i} \neq 0\\ 
	1/n & \text{otherwise} .
	\end{cases}
	\end{equation}
	Then, when $N'<\frac{N \pi_{\min}}{2}$, $\forall \epsilon>0$, let $\tilde{\epsilon} = \min \curlybrackets{\pi_{\min}/2 - N'/N, \epsilon}$, we have
	\begin{equation}
	\prob \parenthesesbig{\norm{\wht{\vP} - \vP} \leq 4 \pi_{\min}^{-1} \norm{\vP} (\epsilon + 1.5N'/N)} 
	\geq  
	1-24 n \tau_* \log(\tilde{\epsilon}^{-1}) 
	\exp \parenthesesbig{- \frac{N}{100 \tau_* \pi_{\max} \log(\tilde{\epsilon}^{-1}) \tilde{\epsilon}^{-2} }} .
	\end{equation}
\end{lemma}
\begin{proof}
	For now, assume $\epsilon<\pi_{\min}/2 - N'/N$. Let $\wht{\vpi}_0, \wht{\vF}_0, \wht{\vP}_0$ be defined the same as Lemma \ref{lemma_VanillaConcentrationofMC} and \ref{lemma_ConcentrationofMC} with the unperturbed trajectory $X_{0:N}$, then according to the proof of Lemma \ref{lemma_ConcentrationofMC}, we have
	\begin{multline} \label{eq_44}
	\prob \parenthesesbig{\norm{\wht{\vF}_0 - \vF} \leq \epsilon, \norm{\wht{\vpi}_0 - \vpi}_\infty \leq \epsilon, \norm{\wht{\vP}_0 - \vP} \leq 4 \pi_{\min}^{-1} \norm{\vP} \epsilon } 
	\geq \\
	1 - 24 n \tau_* \log(\epsilon^{-1}) \exp \parenthesesbig{- \frac{N}{100 \tau_* \pi_{\max} \log (\epsilon^{-1}) \epsilon^{-2}}} .
	\end{multline}
	
	Let $\wht{\vpi}(i) = \frac{1}{N} \sum_{t=1}^{N} \indicator \curlybrackets{\wht{X}_{t-1}=i}, \wht{\vF}(i,j) = \frac{1}{N} \sum_{t=1}^{N} \indicator \curlybrackets{\wht{X}_{t-1}=i, \wht{X}_t=j}$. Then, for any $i \in [n]$, we can see
	\begin{equation}
	\begin{split}
	\wht{\vpi}(i) 
	&= \frac{1}{N} \sum_{t=1}^N \indicator \curlybrackets{\wht{X}_{t-1}=i} \\
	&= \frac{1}{N} \parenthesesbig{
		\sum_{t=1}^N \indicator \curlybrackets{{X}_{t-1}=i} + \sum_{t=1}^N \indicator \curlybrackets{{X}_{t-1}\neq i, \wht{X}_{t-1}=i} - \sum_{t=1}^N \indicator \curlybrackets{{X}_{t-1} = i, \wht{X}_{t-1}\neq i}
	} \\
	&= \wht{\vpi}_0(i) + \frac{1}{N} \parenthesesbig{
		\sum_{t=1}^N \indicator \curlybrackets{{X}_{t-1}\neq i, \wht{X}_{t-1}=i} - \sum_{t=1}^N \indicator \curlybrackets{{X}_{t-1} = i, \wht{X}_{t-1}\neq i}
	}, 
	\end{split}
	\end{equation}
	which gives
	\begin{equation}
	|\wht{\vpi}(i) - \wht{\vpi}_0(i)| \leq \frac{N'}{N} .
	\end{equation}
	Then, with probability no less than the bound given in \eqref{eq_44}, we have
	\begin{equation} \label{eq_45}
	\wht{\vpi}(i) \geq \wht{\vpi}_0(i) - \frac{N'}{N} \geq \vpi(i) - \epsilon - \frac{N'}{N} \geq \frac{\pi_{\min}}{2} >0 ,
	\end{equation}
	\begin{equation} \label{eq_47}
	|\wht{\vpi}(i) - \vpi(i)| \leq |\wht{\vpi}(i) - \wht{\vpi}_0(i)| + |\wht{\vpi}_0 (i) - \vpi(i) | \leq \frac{N'}{N} + \epsilon ,
	\end{equation}
	and
	\begin{equation} \label{eq_46}
	\begin{split}	
	\norm{\wht{\vF} - \vF}
	& \leq \norm{\wht{\vF} - \wht{\vF}_0} + \norm{\wht{\vF}_0 - \vF} \\
	& \leq \norm{\wht{\vF} - \wht{\vF}_0}_F + \epsilon \\
	& = \frac{1}{N} \sqrt{ \sum_{i,j\in [n]} \parenthesesbig{
			\sum_{t=1}^{N} \indicator \curlybrackets{\wht{X}_{t-1}=i, \wht{X}_{t}=j} 
			-
			\sum_{t=1}^{N} \indicator \curlybrackets{{X}_{t-1}=i, {X}_{t}=j}
		}^2} + \epsilon \\
	& \leq \frac{1}{N} \sqrt{ \parenthesesbig{ \sum_{i,j\in [n]} 
			\left| \sum_{t=1}^{N} \indicator \curlybrackets{\wht{X}_{t-1}=i, \wht{X}_{t}=j} 
			-
			\sum_{t=1}^{N} \indicator \curlybrackets{{X}_{t-1}=i, {X}_{t}=j} \right|
		}^2} + \epsilon \\
	& \overset{\text{(i)}}{\leq} \frac{1}{N} \sqrt{(2 N')^2} + \epsilon \\
	& = \frac{2N'}{N} + \epsilon ,
	\end{split}
	\end{equation}
	where (i) holds since $N'$ perturbations in the trajectory can at most ruin $2N'$ transition pairs in total.
	
	Since \eqref{eq_45} guarantees for any $i\in [n], \wht{\vpi}(i) > 0$, similar to the derivation in \eqref{eq_24}, we have
	\begin{equation}
	\norm{\wht{\vP} - \vP} \leq \parenthesesbig{\min_i \wht{\vpi}(i)}^{-1} \norm{\wht{\vF} - \vF} + \frac{\max_i |\wht{\vpi}(i) - \vpi(i)|}{\min_i \wht{\vpi}(i)} \norm{\vP} .
	\end{equation}
	Combining \eqref{eq_45}, \eqref{eq_47} and \eqref{eq_46}, we have
	\begin{equation}
	\begin{split}
	\norm{\wht{\vP} - \vP}
	& \leq 2 \pi_{\min}^{-1} (\frac{2N'}{N} + \epsilon) + 2 \pi_{\min}^{-1} (\frac{N'}{N} + \epsilon) \norm{\vP} \\
	& \leq 2 \pi_{\min}^{-1} (\frac{2N'}{N} + \epsilon) \norm{\vP} + 2 \pi_{\min}^{-1} (\frac{N'}{N} + \epsilon) \norm{\vP} \\
	& \leq 4 \pi_{\min}^{-1} (\frac{3N'}{2N} + \epsilon) \norm{\vP} .
	\end{split}
	\end{equation}
	So,
	\begin{equation}
	\prob \parenthesesbig{\norm{\wht{\vP} - \vP} \leq 4 \pi_{\min}^{-1} \norm{\vP} (\epsilon + 1.5N'/N)} 
	\geq  
	1-24 n \tau_* \log({\epsilon}^{-1}) 
	\exp \parenthesesbig{- \frac{N}{100 \tau_* \pi_{\max} \log({\epsilon}^{-1}) {\epsilon}^{-2} }} .
	\end{equation}
	Finally, we could remove the restriction $\epsilon<\pi_{\min}/2 - N'/N$. For any $\epsilon>0$, let $\tilde{\epsilon} = \min \curlybrackets{\epsilon, \pi_{\min}/2- N'/N}$, then
	\begin{equation}
	\begin{split}
	\prob \parenthesesbig{\norm{\wht{\vP}_0 - \vP} \leq 4 \pi_{\min}^{-1} \norm{\vP} (\epsilon + 1.5N'/N) } 
	&\geq
	\prob \parenthesesbig{\norm{\wht{\vP}_0 - \vP} \leq 4 \pi_{\min}^{-1} \norm{\vP} (\tilde{\epsilon} + 1.5N'/N)} \\
	&\geq 	
	1-24 n \tau_* \log({\tilde{\epsilon}}^{-1}) 
	\exp \parenthesesbig{- \frac{N}{100 \tau_* \pi_{\max} \log(\tilde{\epsilon}^{-1}) \tilde{\epsilon}^{-2} }} ,
	\end{split}
	\end{equation}	
	which concludes the proof.	
\end{proof}

\section{Proofs for Main Theorems} \label{sctn_mainproof}
In this section, we list the proofs for the main theorems appear in the paper. The main idea of the proof for Theorem \ref{thrm_main_MR} is inspired by \cite{zhang2018state}, which is developed for discrete Markov chains. We generalize the work in \cite{zhang2018state} to the case when Markov matrix is not exactly aggregatable and there are perturbations in the Markov chain trajectory, i.e., the mode sequence is estimated from the observation trajectory of an underlying switched system rather than being directly observed. 

\subsection{Proof for Theorem \ref{thrm_MCDiff}}
\begin{proof}
	From Section 3.6 in \cite{cho2001comparison}, one can easily obtain \eqref{eq_40}. For \eqref{eq_41}, we have
	\begin{equation}
		\norm{\vpi_t - \tilde{\vpi}_t}_1 \leq \norm{\vpi_t - \vpi}_1 + \norm{\tilde{\vpi}_t-\tilde{\vpi}}_1 + \norm{\vpi - \tilde{\vpi}}_1 .
	\end{equation}
	By Markov convergence theorem \cite{levin2017markov}, we could upper bound the first two terms and finish the proof.
\end{proof}

\subsection{Proof for Lemma \ref{lemma_NoMistakeCondition}}
\begin{proof}
	Suppose pair $\curlybrackets{y_t, \vphi_t}$ is generated by $\vw_i$, i.e. $y_t = \vw_i^\T \vphi_t + n_t$. Then based on Line \ref{algline_1} in Algorithm \ref{Alg_1}, $\wht{X}_t = X_t$ when $|y_t - \vw_i^\T \vphi_t| < |y_t - \vw_j^\T \vphi_t|, \forall j\neq i$, which is equivalent to
	\begin{equation}\label{eq_43}
		|n_t| < |(\vw_i-\vw_j)^\T \vphi_t + n_t| .
	\end{equation}
	A sufficient condition to guarantee \eqref{eq_43} is $(\vw_i - \vw_j)^\T \vphi_t > 2 n_{\max}, \forall j\neq i$.
\end{proof}

\subsection{Proof for Theorem \ref{thrm_main_MR}}
We first consider the case when there is no estimation error in empirical Markov matrix $\wht{\vP}$, i.e. $\wht{\vP} = \vP$, then generalize this to Theorem \ref{thrm_main_MR}.
\begin{lemma} \label{lemma_ClusteronTrueMatrix}
	Assume: (i) the framework in Section \ref{subsec_1} holds; (ii) in Algorithm \ref{Alg_1}, $\wht{\vP} = \vP$, i.e. the clustering is applied to the true Markov matrix; (iii) $\curlybrackets{\wht{\Omega}_1, \dots, \wht{\Omega}_r}$ is a  $(1+\epsilon_1)$ solution to the k-means problem in Algorithm \ref{Alg_1}. Then, if
	\begin{equation}
		\norm{\vDelta} \leq \frac{\sigma_r(\bar{\vP})}{8 \sqrt{(2+\epsilon_1) r}} \sqrt{\frac{|\Omega_{(r)}|}{ |\Omega_{(1)}|}  + 1} , 
	\end{equation}
	we have
	\begin{equation}
	MR(\wht{\Omega}_1, \wht{\Omega}_2, \dots, \wht{\Omega}_r)
	\leq 
	64(2+\epsilon) r \frac{\norm{\vDelta}^2}{\sigma_r(\bar{\vP})^2}.
	\end{equation}
\end{lemma}
\begin{proof}
	From conditions in Lemma \ref{lemma_ClusteronTrueMatrix}, we see $rank(\bar{\vP})=r$. Let $\bar{\vP} = \bar{\vU}_r \bar{\vSigma}_r \bar{\vV}_r^\T$ be the SVD of $\bar{\vP}$, which only preserves the first $r$ singular value components, so $\bar{\vU}_r \in \dm{n}{r}, \bar{\vSigma}_r \in \dm{r}{r}, \bar{\vV}_r \in \dm{n}{r}$. Recall $\vU_r$ defined in the algorithm contains the first $r$ left singular vectors of $\vP$. We define $\vQ = \arg \min_{\vO \in O(r)} \norm{\vU_r - \bar{\vU}_r \vO }_F^2$, where $O(r)$ denotes the orthogonal group with dimension $r$. Then, we have
	\begin{equation} \label{eq_16}
	\begin{split}
	\norm{\vU_r - \bar{\vU}_r \vQ}_F &\overset{\text{(i)}}{\leq} \sqrt{2} \norm{\sin \vTheta (\vU_r, \bar{\vU}_r)}_F \\
	& \leq \sqrt{2r} \norm{\sin \vTheta (\vU_r, \bar{\vU}_r)} \\
	& \overset{\text{(ii)}}{\leq} \frac{2 \sqrt{2r} \norm{\bar{\vP} - \vP}}{\sigma_r (\bar{\vP})} \\
	& = \frac{2 \sqrt{2r} \norm{\Delta}}{\sigma_r (\bar{\vP})} ,
	\end{split}
	\end{equation}
	where (i) follows from Lemma \ref{lemma_subspacedistance}; (ii) follows from Lemma \ref{lemma_AdvancedWedinBound}.
	
	Also note that since $\vQ$ is orthogonal, for all $ i \in \Omega_k, j \in \Omega_l$ we have
	\begin{equation} \label{eq_14}
	\begin{split}
	& \norm{( \bar{\vU}_r \vQ ) ( i,: ) - ( \bar{\vU}_r \vQ ) ( j,: )}\\
	= & \norm{[\bar{\vU}_r( i,: ) - \bar{\vU}_r ( j,: )] \vQ} \\
	= & \norm{\bar{\vU}_r( i,: ) - \bar{\vU}_r ( j,: )} \\
	= &
	\begin{cases}
	0 & \text{ if } k=l \\ 
	\sqrt{\frac{1}{|\Omega_k|} + \frac{1}{|\Omega_l|}}& \text{ if } k \neq l 
	\end{cases} , 
	\end{split}
	\end{equation}
	where the last line follows from Lemma \ref{lemma_3}. Recall $\curlybrackets{\Omega_1, \dots, \Omega_r}$ is the partition of rows of $\bar{\vP}$ such that rows within the same cluster are the same. We can see in matrix $\bar{\vU}_r \vQ$, rows are the same if corresponding rows in $\bar{\vP}$ are in the same cluster, while different if corresponding rows in $\bar{\vP}$ are in different clusters. So, we can claim that the matrix $\bar{\vU}_r \vQ$ carries the same aggregation information as $\bar{\vP}$. Because of this, together with the fact that we apply k-means to $\vU_r$, we can apply Lemma \ref{lemma_kmeans} by replacing $\curlybrackets{\bar{\vU}, \vU}$ in Lemma \ref{lemma_kmeans} with $\curlybrackets{\bar{\vU}_r \vQ, \vU_r}$.
	
	To make condition \eqref{eq_15} in Lemma \ref{lemma_kmeans} hold, based on \eqref{eq_14}, we can pick
	\begin{equation}
	\delta_k = \sqrt{\frac{1}{|\Omega_k|} + \frac{1}{|\Omega_{(1)}|}} \qquad \forall k \in [r] .
	\end{equation}
	
	To make condition \eqref{eq_12} in Lemma \ref{lemma_kmeans} hold, with this choice of $\delta_k$, it suffices to guarantee for all $ k \in [r]$
	\begin{equation}
	\frac{|\Omega_k|}{|\Omega_{(1)}|}
	>
	8(2+\epsilon) \norm{\bar{\vU}_r \vQ - \vU_r}_F^2 - 1
	\end{equation}
	which, by applying \eqref{eq_16}, can be guaranteed by the following condition:
	\begin{equation} \label{eq_17}
	\norm{\vDelta} \leq \frac{\sigma_r(\bar{\vP})}{8 \sqrt{(2+\epsilon_1) r}} \sqrt{\frac{|\Omega_{(r)}|}{ |\Omega_{(1)}|}  + 1}.
	\end{equation}

	Therefore, with \eqref{eq_17}, and according to Lemma \ref{lemma_kmeans}, we can claim states in set $G$ defined in Lemma \ref{lemma_kmeans} can be correctly aggregated under relabeling invariance. Moreover,
	\begin{equation}
	\begin{split}
	MR(\hat{\Omega}_1, \dots, \hat{\Omega}_r)
	\leq & \sum_{k=1}^{r} \frac{|S_k|}{|\Omega_k|} \\
	\leq & \sum_{k=1}^{r} |S_k| \delta_k^2 \\
	\overset{\text{(i)}}{\leq} & 8(2+\epsilon) \norm{\bar{\vU}_r \vQ - \vU_r}_F^2 \\
	\overset{\text{(ii)}}{\leq} & \frac{64(2+\epsilon) r \norm{\Delta}^2}{\sigma_r (\bar{\vP})^2} ,
	\end{split}
	\end{equation}
	where $S_k$ is defined in Lemma \ref{lemma_kmeans}, (i) follows from \eqref{eq_13} and (ii) follows from \eqref{eq_16}.
\end{proof}

Now, with the analyses under the assumption that no estimation error in $\wht{\vP}$ exists, we go back to the proof for the general case, i.e. Theorem \ref{thrm_main_MR}.

\begin{proof}[Proof for Theorem \ref{thrm_main_MR}]
	Let $\wht{\vDelta} = \vDelta + (\wht{\vP} - \vP)$, then we can see $\wht{\vP} = \bar{\vP} + \wht{\vDelta}$. By applying Lemma \ref{lemma_ClusteronTrueMatrix} to $\wht{\vP}$ and $\wht{\vDelta}$, we can see when
	\begin{equation}\label{eq_20}
	\norm{\wht{\vDelta}}^2 
	\leq 
	\frac{\sigma_r (\bar{\vP})^2}{64(2+\epsilon_1)r} \parenthesesbig{ \frac{|\Omega_{(r)}|}{ |\Omega_{(1)}|}  + 1 } ,
	\end{equation}
	we have
	\begin{equation} \label{eq_29}
	MR(\wht{\Omega}_1, \wht{\Omega}_2, \dots, \wht{\Omega}_r)
	\leq 
	64(2+\epsilon_1) r \frac{\norm{\wht{\vDelta}}^2}{\sigma_r(\bar{\vP})^2} .
	\end{equation}
	
	To guarantee \eqref{eq_20}, by triangle inequality, it suffices to ensure
	\begin{equation} \label{eq_27}
	\norm{\wht{\vP} - \vP} 
	\leq 
	\frac{\sigma_r(\bar{\vP})}{8 \sqrt{(2+\epsilon_1) r}} \sqrt{\frac{|\Omega_{(r)}|}{ |\Omega_{(1)}|}  + 1} - \norm{\vDelta} .
	\end{equation}	
	Now, for all $ \epsilon_2 > 0$, let $\tilde{\epsilon}_2 = \min \curlybracketsbig{\epsilon_2, \frac{\pi_{\min}}{2}-\eta, \frac{\pi_{\min}}{4 (\sigma_1(\bar{\vP}) + \norm{\vDelta})} \parenthesesbig{\frac{\sigma_r(\bar{\vP})}{8 \sqrt{(2+\epsilon_1) r}} \sqrt{\frac{|\Omega_{(r)}|}{ |\Omega_{(1)}|}  + 1} - \norm{\vDelta}} }$. By Lemma \ref{lemma_ConcentrationofMC_withMistake}, we can see when $N \geq 200 \tau_* \pi_{\max} \log (\tilde{\epsilon}_2^{-1}) \tilde{\epsilon}_2^{-2} \squarebrackets{\log(24 n \tau_*) + \log (\log (\tilde{\epsilon_2}^{-1}))}$, with probability no less than
	\begin{equation} \label{eq_26}
	1 - \exp \parenthesesbig{- \frac{N}{200 \tau_* \pi_{\max} \log(\tilde{\epsilon}_2^{-1}) \tilde{\epsilon}_2^{-2} }}, 
	\end{equation}
	we have
	\begin{equation} \label{eq_28}
	\norm{\wht{\vP} - \vP} \leq 4 \pi_{\min}^{-1} \norm{\vP} (\epsilon_2 + 1.5 \eta).
	\end{equation}	
	%	\begin{equation} \label{eq_28}
	%	\prob \parenthesesbig{\norm{\wht{\vP} - \vP} \leq 4 \pi_{\min}^{-1} \norm{P} \epsilon_2} 
	%	\geq  
	%	1-24 n \tau_* \log(\tilde{\epsilon}^{-1}) 
	%	\exp \parenthesesbig{- \frac{N}{100 \tau_* \pi_{\max} \log(\tilde{\epsilon}_2^{-1}) \tilde{\epsilon}_2^{-2} }}
	%	\end{equation}
	By the choice of $\tilde{\epsilon}_2$, we also can see \eqref{eq_26} gives the probability lower bound on the occurrence of \eqref{eq_27}, which further lower bounds the occurrence probability of \eqref{eq_20}.
	
	Finally, plugging \eqref{eq_28} into \eqref{eq_29}, we have
	\begin{equation}
	\begin{split}
	MR(\wht{\Omega}_1, \wht{\Omega}_2, \dots, \wht{\Omega}_r)
	&\leq  64(2+\epsilon_1) r \parenthesesbig{\frac{\norm{\wht{\vDelta}}}{\sigma_r(\bar{\vP})}}^2 \\
	&\leq  64(2+\epsilon_1) r \parenthesesbig{\frac{\norm{\vDelta} + \norm{\wht{\vP} - \vP} }{\sigma_r(\bar{\vP})}}^2 \\
	&\leq  64(2+\epsilon_1) r \parenthesesbig{\frac{\norm{\vDelta} + 4 \pi_{\min}^{-1} \norm{\vP} (\epsilon_2 + 1.5 \eta) }{\sigma_r(\bar{\vP})}}^2 \\
	&\leq  64(2+\epsilon_1) r \parenthesesbig{\frac{\norm{\vDelta}}{\sigma_r(\bar{\vP})} + \frac{4  (\epsilon_2 + 1.5 \eta)  (\norm{\vDelta} + \norm{\bar{\vP}})}{\pi_{\min} \sigma_r(\bar{\vP})} }^2, \\
	\end{split}
	\end{equation}
	which concludes the proof.
\end{proof}

\subsection{Proof for Theorem \ref{thrm_main_P}}
\begin{proof}
	By triangle inequality, we see
	\begin{equation}
		\norm{\vP - \wtd{\vP}}_\infty \leq \norm{\vP - \wht{\vP}}_\infty + \norm{\wht{\vP} - \wtd{\vP}}_\infty .
	\end{equation}
	Assume $i \in \wht{\Omega}_s$, from Line \ref{algline_5} in Algorithm \ref{Alg_1}, we can see $\wtd{\vP}(i,:)$ is a convex combination of $\wht{\vP}(j,:), \forall j \in \Omega_s $. Therefore,
	\begin{equation}
	\begin{split}
	&\norm{\wht{\vP}(i,:)^\T - \wtd{\vP}(i,:)^\T}_1 \\
	\leq & \max_{k,j \in \wht{\Omega}_s} \norm{\wht{\vP}(k,:)^\T - \wht{\vP}(j,:)^\T }_1 \\
	\leq & \max_{k,j \in \wht{\Omega}_s} \norm{\wht{\vP}(k,:)^\T - {\vP}(k,:)^\T }_1 + \norm{\wht{\vP}(j,:)^\T - {\vP}(j,:)^\T }_1 + \norm{{\vP}(k,:)^\T - {\vP}(j,:)^\T }_1 \\
	\leq & 2 \norm{\wht{\vP} - \vP}_\infty + \max_{k,j \in \wht{\Omega}_s} \norm{\bar{\vP}(k,:)^\T - {\vP}(k,:)^\T }_1 + \norm{\bar{\vP}(j,:)^\T - {\vP}(j,:)^\T }_1 + \norm{\bar{\vP}(k,:)^\T - \bar{\vP}(j,:)^\T }_1 \\
	\leq & 2 \norm{\wht{\vP} - \vP}_\infty + 2\norm{\vDelta}_\infty ,
	\end{split}
	\end{equation}
	which gives
	\begin{equation}
		\norm{\vP - \wtd{\vP}}_\infty \leq 3 \norm{\wht{\vP} - \vP}_\infty + 2\norm{\vDelta}_\infty .
	\end{equation}
	Plugging in $\norm{\wht{\vP} - \vP} \leq 4 \pi_{\min}^{-1} \norm{\vP} (\epsilon_2 + 1.5 \eta)$ derived in the proof of Theorem \ref{thrm_main_MR}, we conclude the proof.
\end{proof}

\end{document}